%% file: ArXiV_SWIPT_Journal.tex
\documentclass[12pt,draftclsnofoot,onecolumn]{IEEEtran}
\input{packagesDK} 
\input{newcomDK} 
\input{acronsDK} 

\usepackage{dblfloatfix}    

\IEEEoverridecommandlockouts

\def\@IEEEfigurecaptionsepspace{\vskip\abovecaptionskip\relax}%

\usepackage{bbm}
\usepackage{amsmath, amsthm, amssymb}
\newtheorem{proposition}{Proposition}

\newcounter{mytempeqncnt}

\newcommand{\subparagraph}{} 

\begin{document}
\title{Latency-Aware Multi-antenna SWIPT System with Battery-Constrained Receivers}
\author{
\IEEEauthorblockN{Dileep~Kumar,~\IEEEmembership{Student~Member,~IEEE,} Onel L. Alcaraz López,~\IEEEmembership{Member,~IEEE,} Satya Krishna Joshi,~\IEEEmembership{Member,~IEEE,} and~Antti~T\"{o}lli~\IEEEmembership{Senior~Member,~IEEE}}
\thanks{This work was supported by the European Commission in the framework of the H2020-EUJ-02-2018 project under grant no.~815056 ({5G}-Enhance) and Academy of Finland under grants no.~346208~(6Genesis Flagship).}
\thanks{This article was presented in part at the IEEE Int. Symp. Pers., Indoor, Mobile Radio Commun. (PIMRC), Sep 2021~\cite{dileep_PIMRC}, and in parts at the Int. Symp. Wireless Commun. Systems (ISWCS), Sep 2021~\cite{dileep_ISWCS}.  \textit{(Corresponding~author: Dileep~Kumar)} 
}
\thanks{Authors are with Centre for Wireless Communications (CWC), University of Oulu, FIN-90014 Oulu, Finland. (e-mail: \{dileep.kumar, onel.alcarazlopez, satya.joshi, antti.tolli\}@oulu.fi).}
}

\makeatletter
\let\old@ps@headings\ps@headings
\let\old@ps@IEEEtitlepagestyle\ps@IEEEtitlepagestyle
\def\confheader#1{%
\def\ps@IEEEtitlepagestyle{%
\old@ps@IEEEtitlepagestyle%
\def\@oddhead{\strut\hfill#1\hfill\strut}%
\def\@evenhead{\strut\hfill#1\hfill\strut}%
}%
\ps@headings%
}
\makeatother

\confheader{%
\indent \footnotesize{Accepted for publication in IEEE Transactions on Wireless Communications in Oct. 2022. (DOI: 10.1109/TWC.2022.3215864)}  
}

\maketitle


%

\begin{abstract}



Power splitting~(PS) based simultaneous wireless information and power transfer (SWIPT) is considered in a multi-user multiple-input-single-output broadcast scenario. Specifically, we focus on jointly configuring the transmit beamforming vectors and receive PS ratios to minimize the total transmit energy of the base station under the user-specific latency and energy harvesting (EH) requirements. The battery depletion phenomenon is avoided by preemptively incorporating information regarding the receivers' battery state and EH fluctuations into the resource allocation design. The resulting time-average sum-power minimization problem is temporally correlated, non-convex (including mutually coupled latency-battery queue dynamics), and in general intractable. We use the Lyapunov optimization framework and derive a dynamic control algorithm to transform the original problem into a {sequence of per-time-slot deterministic and independent subproblems}. The latter are then solved via two alternative approaches: i)~semidefinite relaxation combined with fractional programming, and ii)~successive convex approximation. Furthermore, we design a low-complexity closed-form iterative algorithm exploiting the Karush-Kuhn-Tucker optimality conditions for a specific scenario with delay bounded batteryless receivers. Numerical results provide insights on the robustness of the proposed designs to realize an energy-efficient SWIPT system while ensuring latency and EH requirements in a time dynamic network.


\end{abstract}


%


%
%
\section{Introduction}
\label{sec:Intro}

Powering and supporting the seamless and autonomous operation of \ac{IoT} deployments is becoming challenging given the exponential increase in the number of ubiquitous and low-power  devices. 
On one hand, wired charging is usually cost-prohibitive and/or inconvenient to deploy anywhere, especially when considering industrial \ac{IoT} deployments. On the other hand, traditional battery-powered solutions struggle with limited battery life and the corresponding replacement problem, which is neither economical nor eco-friendly. 
%
{Unfortunately, the energy scavenging from environmental sources (e.g., solar, wind, thermal, etc) may be also inappropriate (at least as a standalone) for many use cases with stringent \ac{QoS} requirements. This is because the associated energy supply is uncontrollable and can significantly oscillate according to  temporal/geographical/environmental circumstances.
Moreover, the ambient energy scavenging demand an add-on harvesting material and circuit, which in practice limits the form factor reduction to the desired levels for many use cases.
As an appealing alternative, \ac{RF}-based wireless energy transfer (WET) allows efficient on-demand charging of low-power devices over the air, which simplifies servicing and maintenance, promotes form factor reduction and durability increase of the end devices, and thus contributes to the realization of scalable and sustainable wireless networks~\cite{Lopez.2021}.}

Wireless networks natively include wireless information transfer processes, hence WET appears naturally combined with them.
The \ac{SWIPT} technology,  which enables efficient utilization of radio resources by wirelessly charging battery-limited devices while simultaneously conveying useful information~\cite{Meng-HU_HSU-2016, Choi-SWIPT-2016, Clerckx-2019-JSAC}, is an example of it. 
Specifically, the
downlink multi-antenna broadcast \ac{SWIPT} system has attracted much attention from the research community and has been widely studied in the literature. 
As an example, the early pioneering work~\cite{Zhang-MIMO-Braoadcasting-2013} characterizes the underlying rate-energy trade-offs of a broadcast \ac{SWIPT} system that includes a multi-antenna \ac{BS} communicating with several \acp{UE} in downlink. Therein, each \ac{UE} may perform both \ac{EH} and \ac{ID} functions on the received signal, i.e., either by applying the time-switching (TS)~\cite{Zhang-MIMO-Braoadcasting-2013} or \ac{PS}~\cite{Zhang-MIMO-Braoadcasting-2013, Zhou-Zhang-2013} protocol. 
{Meanwhile, the authors in~\cite{Rv_Delay-aware} investigate the problems of maximizing the effective capacity and energy efficiency by considering both \ac{PS} and TS architectures subject to a fixed minimum harvested energy and fixed \ac{QoS} delay requirements for receivers.} 
Although a {TS} scheme simplifies the receiver design, it hinders the full exploitation of radio resources since the wireless information and WET processes are not really simultaneous in time. Thus, motivating the use of \ac{PS} schemes which generally {achieve higher spectral efficiency~\cite{Xiao-JSAC-2016, Rv_Delay-aware} and energy efficiency~\cite{Rv_Delay-aware}}. A comprehensive overview of the aforementioned PS and TS schemes and other less popular \ac{SWIPT} schemes, such as the so-called spatial switching~(SS) and integrated receiver~(IR), is provided in~\cite{Lu.2015}.

It is worth highlighting that interference substantially reduces the data rate achievable at each link, thus it has an adverse impact on \ac{ID} performance. Meanwhile,
in terms of WET, the interference is not a foe but a friend since the power of  interfering signals can easily be transformed into useful energy, and conceivably used for recharging  \ac{EH} \acp{UE}~\cite{Ikhlef-ComLet-2014, Xu-Zhang-TSP-2014, Park-TWC-2014,Gautam_-2021_Journal}. 
Due to the simultaneous multi-user \ac{RF} \ac{EH}, the SWIPT beamformer design inherently has a multicast structure.
Note that finding a good balance between \ac{EH} and \ac{ID} is key for an efficient \ac{SWIPT} implementation. 
For scenarios with separated \ac{EH} and \ac{ID} receivers, the authors of~\cite{Ikhlef-ComLet-2014} study the minimization of the worst-case \ac{MSE} of the received information signals under transmit power and receive EH constraints. In~\cite{Xu-Zhang-TSP-2014}, the transmit beamformer is designed to maximize the weighted sum of users' harvested power subject to per-user \ac{SINR} requirements. Meanwhile, authors in~\cite{Park-TWC-2014} characterize the achievable rate-energy trade-off in a SWIPT multi-user MIMO interference channel. 
In~\cite{Qingjiang-TWC-SWIPTKey_2014}, a joint beamforming and \ac{PS} optimization is carried out to minimize the BS transmit power, under fixed \ac{SINR} and \ac{EH} \ac{QoS} constraints per user, with the help of the \ac{SDR} framework. 
This work is extended in~\cite{Shi-SWIPT-2014}, wherein the authors develop a low-complexity and computationally efficient solution via \ac{SOCP}. 
Authors of~\cite{Zhang-TVT-2017, Dong-ICC-2016} also extend~\cite{Qingjiang-TWC-SWIPTKey_2014} by studying and optimizing the transceiver design under  \ac{MSE} and \ac{EH} constraints. Meanwhile, a joint transmit power allocation and receiver \ac{PS} mechanism is proposed in~\cite{Baosheng-Zhang-ComLet-2016} to maximize the minimum \ac{SINR} at the \acp{UE}. Assuming co-located \ac{EH} and \ac{ID} receivers, the minimization of the sum power consumption is studied in~\cite{Gautam_-2021_Journal, Timotheou-TWC-2014, Zong-TWC-2016}, while~\cite{Qingjiang-TSP-2016} considers the energy efficiency maximization, where beamforming vectors and \ac{PS} ratios are jointly designed subject to a fixed minimum \ac{EH} and \ac{SINR} constraints per~\ac{UE}.

All the aforementioned works mainly focus on a downlink multi-user \ac{SWIPT} system under fixed users' \ac{QoS} requirements. Moreover, they assume that the instantaneous harvested energy is always sufficient to support the devices' computational and communication operations. However, the energy required to decode the received data increases with  the downlink transmitted rate~\cite{Javier-Battery-2014, Zhirui-Battery-2017}, and thus the above assumption may not be always satisfied.
Specifically, if the energy required for the information decoding operation surpasses the current energy availability, 
the UE could turn off, or at least interrupt the ID operations,
thus causing a service outage. This critical problem 
has lately attracted significant research effort in diverse setups.
%
%
For  information-only transfer systems with limited energy availability at the receivers, authors of~\cite{Javier-Battery-2014} investigate the downlink transmission strategies and provide the feasible maximum supported data rates, which are mainly restricted by the receivers' instantaneous battery levels. The previous work is extended to coordinated multi-cell systems in~\cite{Zhirui-Battery-2017}. {Meanwhile, authors of~\cite{Zhaoxia-BatteryConf-2016, Sun-BatteryConf-2018, Zhirui-Battery-2020} aim to maximize the network sum rate while considering the battery depletion problem, for which they establish a limited receiver battery constraint.  In~\cite{Zhaoxia-BatteryConf-2016, Zhirui-Battery-2020}, the solutions are only based on the instantaneous states of the batteries, while~\cite{Sun-BatteryConf-2018} applies the Lyapunov optimization framework to design an online dynamic control algorithm.

The Lyapunov based solution has been widely used in many practical applications (as,  e.g., in \cite{Rv_Buffer-Aided, Rv_Optimal_Online, Rv_Energy_Efficiency, Rv_Resource_and_Task}) to transform the long-term average stochastic problem into a series of per-time-slot deterministic optimization subproblems, which do not require any a priori information on the system uncertainties, while efficiently capturing the non-stationary evolution of queue backlogs per user. In~\cite{Rv_Buffer-Aided}, the authors propose a \ac{SWIPT}-enabled mechanism to increase the average secrecy throughput by considering the data and energy buffering at the wirelessly powered relays. Meanwhile, the authors of~\cite{Rv_Optimal_Online} design an online adaptive resource allocation policy to maximize the fairness among different mobile clients, while the trade-off between response delay and energy efficiency is studied in~\cite{Rv_Energy_Efficiency}. Finally, authors in \cite{ Rv_Resource_and_Task} propose a task scheduling algorithm for SWIPT systems to minimize the energy consumption, while satisfying a minimum average data and task performing rate. }

%


\vspace*{-10pt}
\subsection{Contributions}
\label{subsec:contribution}
\vspace*{-5pt}
Although the state-of-the-art research on SWIPT-enabled systems is vast, the increasingly stringent system QoS requirements, low energy availability, and battery dynamics over time still require extensive study, which motivates our research in this work. 
Specifically, we provide the joint optimization of transmit beamforming vectors and receive \ac{PS} ratios that concurrently satisfies the user-specific latency and \ac{EH} requirements of a \ac{PS}-based \ac{SWIPT} system. 
Our proposed radio resource allocation strategies preemptively incorporate information related to the users' battery state and EH fluctuations while efficiently avoiding the battery depletion phenomenon in a delay bounded time dynamic mobile access network. To the best of our knowledge, this paper, whose preliminary and reduced-length versions can be found in our earlier works~\cite{dileep_PIMRC, dileep_ISWCS}, is the first to jointly consider the interdependence of
latency-battery dynamics for \ac{SWIPT}-enabled systems. The key contributions of this paper are as follows:
\begin{itemize}
    \item We formulate a long-term transmit energy minimization problem subject to users' probabilistic queue backlog constraint and a maximum harvested power requirement. The transmit beamforming vectors and receive \ac{PS} ratios are jointly designed considering explicitly the minimum energy needed to decode downlink transmitted data messages and to support battery-limited receivers' operations. 
    Our proposed radio resource allocation schemes allow meeting the latency and EH requirements at the \acp{UE} while avoiding their battery depletion.
    \item We employ the Lyapunov optimization framework, specifically the drift-plus-penalty function, to transform the time-average problem {into a sequence of per-time-slot deterministic subproblems, while efficiently capturing the non-stationary evolution of queue backlogs}. To circumvent the  temporally correlated user-specific battery and EH constraints, we introduce a virtual spare battery queue for each user. A dynamic control algorithm is then provided for solving the time-correlated and cumbersome average sum-power minimization problem.
    \item The beamformers and PS ratios are coupled in \ac{SINR} and EH expressions, and are not jointly convex. To circumvent this, {two alternative approaches}, namely \ac{SDR} combined with \ac{FP} quadratic transform technique, and a linear Taylor series approximation via \ac{SCA} framework, are adopted. Both proposed methods achieve  efficient solutions, each with specific convergence and complexity characteristics. 
    Specifically, the \ac{SDR}-FP based method provides faster convergence in terms of  the required approximation point updates for the considered setup configuration. In contrast, the \ac{SCA} framework has much lower computational complexity per iteration, and the optimized beamformers are obtained directly. 
    It is worth highlighting that the constrained optimization schemes proposed in~\cite{Qingjiang-TWC-SWIPTKey_2014, Shi-SWIPT-2014, Zhang-TVT-2017, Dong-ICC-2016, Baosheng-Zhang-ComLet-2016, Timotheou-TWC-2014, Zong-TWC-2016, Qingjiang-TSP-2016} can not be used directly here due to the underlying interdependence of latency-battery queue dynamics. Thus, our proposed solutions are significantly more advanced, and provide a {systematic} approach for solving {problems with} mutually coupled and non-linear time-varying constraints. 
    %
    %
    \item A low-complexity and computationally efficient iterative algorithm is proposed for a special scenario where the latency constrained SWIPT-enabled receivers operate without energy storage, and the harvested energy is made immediately available to support their operations~\cite{Qingjiang-TWC-SWIPTKey_2014, Shi-SWIPT-2014}. Specifically, the proposed algorithm exploits the \ac{KKT} optimality conditions, and requires evaluating only closed-form mathematical expressions. This leads to computationally efficient, thus, latency-friendly, implementations.
    %
\end{itemize}
The proposed methods provide insights into the trade-offs between the long-term achievable harvested energy at \acp{UE} and required  transmit energy at the \ac{BS}, while ensuring the {user-specific} latency and minimum battery energy requirements to support the persistent receivers' operations. The simulation results manifest the robustness of the proposed designs to realize energy-efficient \ac{SWIPT} systems for delay bounded applications in a time dynamic mobile access~network.

\vspace*{-5pt}
\subsection{Organization and Notations}
\label{subsec:Organization-Notations}
 \vspace*{-5pt}
The remainder of this paper is organized as follows. In Section~\ref{sec:model}, we describe the system architecture and  optimization problem. Section~\ref{sec:Lyapunov} provides a dynamic control algorithm for the average sum-power minimization problem. In Section~\ref{sec:Prop_Approximation}, we propose the joint beamforming and \ac{PS} optimization. In Section~\ref{subsec:Infinite_Battery_ISWCS}, we study the batteryless scenario. Simulation results are presented in Section~\ref{sec:sim-results}, while conclusions are drawn in Section~\ref{sec:Conclusion}.

\noindent \textit{Notations}: We use italic, boldface lower- and upper-case letters to denote scalars, vectors and matrices, respectively. Notation $\mathbb{C}^{N \! \times \! M}$ represents the space of ${N \! \times \! M}$ complex matrices. The norm and the real part of a complex number is represented with $|\cdot|$ and $\Re\{\cdot\}$, respectively. For an arbitrary-size matrix $\mathbf{X}$, the superscript $\mathbf{X}^{\scriptsize -1}$, $\mathbf{X}\herm$ and $\mathbf{X}\tran$ donates respectively inverse, conjugate transpose and transpose operation. $[\mathbf{X}]_{m,n}$ is the $(m,n)$-{th} element of~$\mathbf{X}$, and $[x]^{\scriptsize +}\!\triangleq\!\mathrm{max}(0,x)$. Notation $\mathbb{E}[\cdot]$ represents the statistical expectation operation.


%
%

\section{System Architecture and Problem Formulation}
\label{sec:model}

\subsection{System model}
\label{subsec:Sys_model}
%
%
We consider a \ac{MU-MISO} \ac{SWIPT} system as shown in Fig.~\ref{fig:System-Model}. Here, a set $\mathcal{K}\!\triangleq\!\{1,2,\ldots,K\}$ of single antenna \acp{UE} is served by a \ac{BS} equipped with $N_t$~transmit antennas in the downlink.
For simplicity, but without loss of generality, we consider a time-slotted frame structure where the slots are normalized to an integer value, i.e., $t\!\in\!\{T_f,2T_f,\ldots \}$ with a duration of~$T_f$ seconds each. Let $\mathbf{h}_{k}(t), \mathbf{f}_{k}(t)\!\in\!\mathbb{C}^{N_t \times 1}$ denote respectively the downlink channel vector and transmit beamforming vector corresponding to the $k$-th \ac{UE} during time slot~$t$. Then, the  signal received at the $k$-th \ac{UE} can be expressed~as
\begin{align}
    \label{eq:Rx-signal}
    y_k(t) & = \mathbf{h}_{k}\herm(t) \mathbf{f}_{k}(t) {d_{k}(t)}  
    + \!\! \sum\limits_{u \in \mathcal{K} \backslash k} \!\! \mathbf{h}_{k}\herm(t) \mathbf{f}_{u}(t) {d_{u}}(t) + {w}_k(t) ,
\end{align}
where ${w}_k\in\mathcal{CN}(0,\sigma^2_k)$ denotes the additive white Gaussian noise~(AWGN) at the receiver, and $d_k$ is the downlink transmitted data symbol. Moreover, we assume independent and normalized data symbols, i.e., $\mathbb{E}\{d_k d_u^*\}=0$ and $\mathbb{E}\{|d_k|^2\}=1,$ $\forall k, u \in \mathcal{K}$ and $k\ne u$.

Similar to~\cite{Zhang-MIMO-Braoadcasting-2013, Zhou-Zhang-2013, Qingjiang-TWC-SWIPTKey_2014, Shi-SWIPT-2014, Zhang-TVT-2017, Dong-ICC-2016, Baosheng-Zhang-ComLet-2016, Gautam_-2021_Journal, Timotheou-TWC-2014, Zong-TWC-2016, Qingjiang-TSP-2016}, we assume that each user~$k\!\in\!\mathcal{K}$ implements \ac{PS} on the received signal~$y_k(t)$ for simultaneous~\ac{ID} and~\ac{EH} as shown in Fig.~\ref{fig:System-Model}. Let $\rho_k(t)\!\in\![0,1]$ represent the \ac{PS} ratio to the \ac{ID} circuit of $k$-th \ac{UE} during time slot~$t$. Then, the portion of the signal split to the \ac{ID} circuit can be expressed~as
\begin{align}
    \label{eq:Rx-signal-ID}
    y_k^{\mathrm{ID}}(t) & = \sqrt{\rho_k(t)} \Bigl( \mathbf{h}_{k}\herm(t) \mathbf{f}_{k}(t) {d_{k}}(t)
         + \!\!  \sum\limits_{u \in \mathcal{K} \backslash k} \!\! \mathbf{h}_{k}\herm(t) \mathbf{f}_{u}(t) {d_{u}}(t) + {w}_k(t) \Bigr) + \tilde{z}_k(t),
\end{align}
where $\tilde{z}_k\in\mathcal{CN}(0,\delta^2_k)$ is the additive noise at the \ac{ID} circuit of $k$-th \ac{UE}. The received \ac{SINR} of $k$-th \ac{UE} is given~by
\begin{align}
    \label{eq:SINR-ID}
     \Gamma_k(t) & = \frac{\rho_k(t) |\mathbf{h}_{k}\herm(t) \mathbf{f}_{k}(t)|^2}{\rho_k(t) \!\! \sum\limits_{u \in \mathcal{K} \backslash k} \! |\mathbf{h}_{k}\herm(t) \mathbf{f}_{u}(t)|^2 + \rho_k(t)\sigma_k^2 + \delta_k^2}.
\end{align}
Meanwhile, the remaining portion of the signal split to the \ac{EH} circuit of $k$-th \ac{UE} during time slot~$t$ is given~by
\begin{align}
    \label{eq:Rx-signal-EH}
    y_k^{\mathrm{EH}}(t) & = \sqrt{1-\rho_k(t)} \Bigl( \mathbf{h}_{k}\herm(t) \mathbf{f}_{k}(t) {d_{k}}(t) 
    + \!\!  \sum\limits_{u \in \mathcal{K} \backslash k} \!\! \mathbf{h}_{k}\herm(t) \mathbf{f}_{u}(t) {d_{u}}(t) + {w}_k(t) \Bigr).
\end{align}
Then, the harvested power at the \ac{EH} circuit of $k$-th UE is given~by
\begin{align}
    \label{eq:power-EH}
     {E}_k^h(t) = \zeta_k\bigl(1-\rho_k(t)\bigr) \Bigl(\sum\limits_{j \in \mathcal{K}} |\mathbf{h}_{k}\herm(t) \mathbf{f}_{j}(t)|^2 + \sigma_k^2 \Bigr),
\end{align}
where $\zeta_k\!\in\![0,1)$ represents the energy conversion efficiency. 
%
%
%
%
{
Note that we consider an EH model, as in~\cite{Zhang-MIMO-Braoadcasting-2013, Zhou-Zhang-2013, Qingjiang-TWC-SWIPTKey_2014, Shi-SWIPT-2014}, where the output direct current~(DC) increases linearly with the input RF power level. However, in practice, the energy conversion efficiency of a practical EH circuit is not constant due to the non-linearities introduced by, e.g., resistance, capacitance, and diode-connected transistors~\cite{Rv_Boshkovska-ComLet-2015}. Note that there is no generic EH model which can capture all practical issues of EH circuit components~\cite{Rv_Yong-TCom-2017}. Therefore, we are considering a broadly recognized, yet simple, linear model for analytical tractability and to facilitate the discussion. 
Indeed, we aim to capture, for the first time, the underlying interdependence of latency-battery queue dynamics for SWIPT-enabled time dynamic systems, for which leveraging a simple EH model is a natural first step. Moreover, the results and the performance trends shown in this paper should still, {approximately}, hold for the moderate input power level regime~\cite[Fig.~4]{Rv_Resource_and_Task}. Considering a more involved \ac{EH} model, which efficiently accounts for the sensitivity and saturation impairments of the non-linear circuit components, is an interesting topic for future extensions of this work.}

\begin{figure}[t]
\centering
\includegraphics[trim=0.02cm 0.0cm 0.02cm 0.02cm, clip, width=0.9\linewidth]{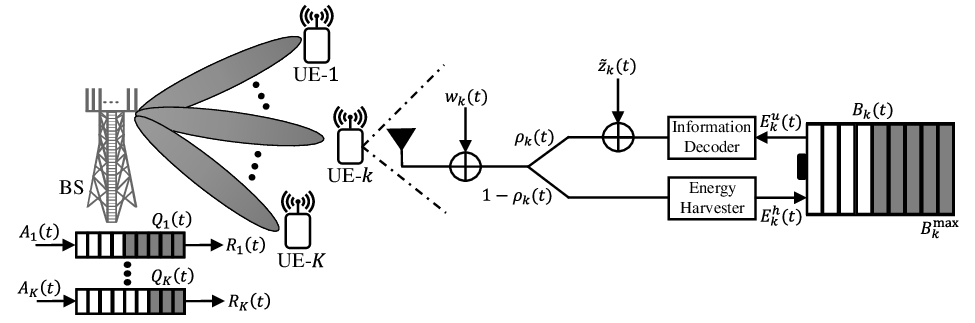} 
\caption{{The {SWIPT} system consists of one BS with  {user-specific} queues, and multiple battery-limited UEs, where each UE implements the PS scheme to perform EH and ID simultaneously.}} 
\label{fig:System-Model}
 \vspace*{-20pt}
\end{figure}
%

\subsection{Queueing \& latency model at network}
\label{subsec:Traffic_Model}

We assume that the \ac{BS} has queue buffers to store the network layer data of \acp{UE}~\cite[Ch. 5]{neely2010stochastic}. Let {$Q_k(t)$ denote the queue backlog of $k$-th \ac{UE} during time slot~$t$}, which evolves~as
\begin{equation}
\label{eq:Queue_Dynamics}
    Q_k(t+1) = \bigl[Q_k(t) - R_k(t) + A_k(t)\bigr]^{+}, \ \  \forall k,
\end{equation}
where $[x]^{+} \triangleq \mathrm{max}(x,0)$, $R_k(t) \triangleq \log_2(1+\gamma_k(t))$ is the downlink rate, and $\gamma_k(t)$ denotes the achievable \ac{SINR} of $k$-th \ac{UE} during time slot~$t$. In~\eqref{eq:Queue_Dynamics}, $A_k(t)$ denotes the downlink data arrival with a mean arrival rate of~$\mathbb{E}[A_k(t)] = \alpha_k,$ $\forall k\in\mathcal{K}$.

According to Little's law, the average delay (or latency) is proportional to long-term time-average queue length as $\lim\limits_{T \to \infty} \! \frac{1}{T} {\textstyle \sum_{t=0}^{T-1}} \mathbb{E}[Q_k(t)]$ \cite[Ch. 1.4]{gross2008fundamentals}. Thus, {we can use the queue backlogs $\{Q_k(t)\}$ as users' latency measures}, and impose the allowable threshold $\{Q_{k}^{\mathrm{th}}\}$ to each time slot. Specifically, we consider a probabilistic constraint on {user-specific queue length}~\cite{Dileep_Globecom2020}, defined~as     
\begin{equation}
    \label{eq:Prob-Queue-Const}
    \mathrm{Pr}\bigl\{ Q_k(t) \geq Q_{k}^{\mathrm{th}} \bigr\} \leq \epsilon, \ \ \forall t,
\end{equation}
where $\epsilon$ is a tolerable probability for delay bound violation.

\subsection{Battery storage \& consumption model at receiver}
\label{subsec:Battery_Model}

\begin{figure}[t]
\centering
\includegraphics[trim=0.02cm 0.0cm 0.02cm 0.02cm, clip, width=0.5\linewidth]{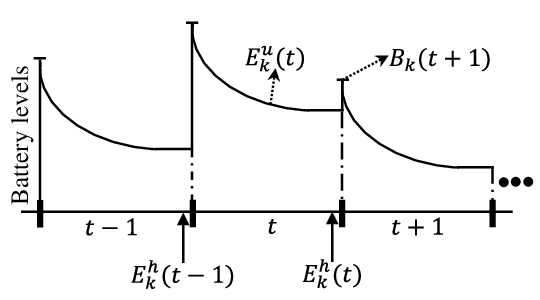} 
\caption{An illustration of the receivers' battery storage and consumption model at the $k$-th UE.} 
\label{fig:Battery_Dynamics}
 \vspace*{-20pt}
\end{figure}

The $k$-th \ac{UE} is equipped with a battery of finite capacity~$B_k^{\mathrm{max}}$. The battery energy level decreases in accordance with the circuit power consumption and the decoding operations of the downlink received data~\cite{Javier-Battery-2014, Zhaoxia-BatteryConf-2016, Zhirui-Battery-2017, Sun-BatteryConf-2018, Zhirui-Battery-2020}. The \ac{EH} circuit allows the \acp{UE} to recharge their battery by collecting the required energy from received \ac{RF} signals (see~\eqref{eq:power-EH}). Similar to~\cite{Zhaoxia-BatteryConf-2016, Sun-BatteryConf-2018}, we assume that the \ac{EH} system adopts the \emph{harvest-store-use} strategy~\cite[Section III]{Meng-HU_HSU-2016}, wherein the harvested energy in a given time slot can be only available in the subsequent time slots, as shown in Fig.~\ref{fig:Battery_Dynamics}. Thus, the battery dynamics $B_k(t+1)$ of the $k$-th \ac{UE} at the beginning of the time slot~$(t+1)$ is given~by
\begin{align}
    \label{eq:Battery_Dynamics}
        B_k(t+1) = \mathrm{min}\Bigl\{B_k^{\mathrm{max}}, \ \bigl[B_k(t)- E_k^u(t)\bigr]^{+} + E_k^h(t)  \Bigr\}, \ \ \forall k, 
\end{align}
where $E_k^h(t)$ is the harvested energy defined in~\eqref{eq:power-EH}. 
Moreover, $E_k^u(t)$ represents the energy spent by the $k$-th \ac{UE} during time slot~$t$, which we model as
\begin{align}
    \label{eq:Energy_Consumption}
     E_k^u(t) = T_f \bigl(P_k^{\mathrm{cir}} + P_k^{\mathrm{dec}}(t)\bigr), \ \ \forall k,
\end{align}
where {$P_k^{\mathrm{cir}}$ is the fixed power consumed by the receiver, including the front-end circuit, and reporting power consumption of the UE to the BS}. In~\eqref{eq:Energy_Consumption}, $P_k^{\mathrm{dec}}(t)$ is the power consumed by the decoding process. Note that the overall energy to be consumed in each time slot must not exceed the energy currently available in the users' battery, i.e., $0\!\leq E_k^u(t) \!\leq B_k(t), \ \forall t,k$. 
%
%
Specifically, by incorporating the receivers' battery state information in the design of beamformers and  \ac{PS} ratios, the \ac{BS} can adequately allocate the required  resources to ensure, i.e., maximum supported  rates~$R_k^{\mathrm{max}}(t)$ and maximum required \ac{EH}~$E_k^{\mathrm{max}}(t)$ for each user at each time slot.

The circuit power~$P_k^{\mathrm{cir}}$  is assumed to be a constant for each user. On the other hand, $P_k^{\mathrm{dec}}(t)$ mainly depends on the transmitted data rates, defined as $P_k^{\mathrm{dec}}(t) \triangleq f(R_k(t)), \ \forall k \in \mathcal{K}$. It is worth highlighting that there is {no}  standardized and extensively recognized model for $f(R_k(t))$ in the research community. {However, there is a consensus about the fact that {the decoding consumption always increases with the downlink transmit rate}.} Herein, we adopt the model presented in~\cite{Javier-Battery-2014}, which states that the energy consumed is  linearly increasing with the downlink transmitted data rates. Specifically, $f(R_k(t))\!\triangleq\!\vartheta_k^d R_k(t)$ where $\vartheta_k^d\!>\!0$ reflects the energy efficiency [Joules/bit] of the decoder implementation at the receiver device. Thus, combined with $P_k^{\mathrm{dec}}(t)$ and expression~\eqref{eq:Energy_Consumption}, the maximum allowable downlink rates~$R_k^{\mathrm{max}}(t)$ of each user~$k$ can be expressed~as\footnote{{Note that the beamformer design and PS strategy proposed in the paper can be extended by adopting the appropriate maximum supported rate~\eqref{eq:Max_Data_Rate} corresponding to any specific implementation of the decoding power consumption model~\cite[Section II]{Javier-Battery-2014}.}} \cite[Section II]{Javier-Battery-2014}
\begin{align}
    \label{eq:Max_Data_Rate}
    R_k(t) & \leq  R_k^{\mathrm{max}}(t) \triangleq \frac{1}{\vartheta_k^d} \biggl( \frac{E_k^u(t) - T_f P_k^{\mathrm{cir}}}{T_f} \biggr), \ \ \forall k.
\end{align}

Note that due to the stochastic nature of the radio channel, there exists a non-negligible probability that the links between the BS and certain \acp{UE} are in poor conditions (i.e., in a deep fading state). Thus, scheduling such \acp{UE} will probably provide little (if any) benefit. However, storing the energy and waiting for better channel gains may  improve the network performance. 
{Therefore, at the beginning of each time slot, the resource allocation process at the \ac{BS} should take into account the user-specific battery charge levels, and accordingly assign the downlink  rates. }
From \eqref{eq:Energy_Consumption}, we can observe that even if the~$k$-th UE is not served at a given scheduling time slot~$t$, its battery will still decrease {according} to $T_fP_k^{\mathrm{cir}}$. 
%
This fixed energy consumption  impacts the resource allocation algorithm for the EH process, i.e., the battery energy should not be instantaneously depleted and there should be always sufficient energy stored in the receiver's battery to support its persistent operations. 
%

The battery storage capacity of the~$k$-th UE is limited by $B_k^{\mathrm{max}}$, and the overflow (or extra) harvested energy will be eventually discarded by the user device. Thus, the harvested energy at each time slot~$t$ can be upper bounded by
\begin{align}
    \label{eq:Max_Energy_Harvested}
    E_k^h(t) \leq E_k^{\mathrm{max}}(t) \triangleq B_k^{\mathrm{max}} - B_k(t), \ \ \forall k.
\end{align}
%

To summarize, the current battery charge level at each \ac{UE} strictly constrains both the maximum supported downlink transmit rates~\eqref{eq:Max_Data_Rate} and the maximum EH requirements~\eqref{eq:Max_Energy_Harvested} at each time slot. 
We assume a standard time division duplex (TDD)-based \ac{CSI} acquisition from reciprocal uplink, where the current battery energy availability is also reported by the \acp{UE}~\cite{Javier-Battery-2014, Zhaoxia-BatteryConf-2016, Zhirui-Battery-2017, Sun-BatteryConf-2018, Zhirui-Battery-2020}.  
Hence, the resource allocation process not only considers \ac{CSI}, but also the 
residual energy in the receivers' battery, 
for the joint design of transmit beamforming vectors and receive \ac{PS} ratios in a delay bounded time dynamic mobile access network.



\subsection{Problem formulation}
\label{subsec:problem}

Our objective is to design a power-efficient resource allocation scheme that satisfies the latency requirement of each user while taking into account the user-specific battery energy limitations.  
Specifically, we jointly optimize transmit beamforming vectors and receive \ac{PS} ratios {in order} to minimize the \ac{BS} average transmit power subject to a probabilistic queue backlog constraint and a maximum harvested power {requirement} per user. Furthermore, the resource allocation strategies are designed by preemptively incorporating the  receivers’ battery and EH fluctuations while efficiently avoiding the battery depletion phenomenon. 
The 
network {average transmit sum-power minimization} 
problem can be formulated~as  
\begin{subequations}
\label{eq:P1-prob}
\begin{align}
	\displaystyle 
	 \underset{\mathbf{f}_k(t), \rho_k(t), \ \forall t}{\mathrm{min}} \ \ & \lim_{T \to \infty} \frac{1}{T} \sum\limits_{t=0}^{T-1}  \sum\limits_{k \in \mathcal{K}} \mathbb{E}\bigl[\|\mathbf{f}_{k}(t)\|^2\bigr]  \label{eq:P1} \\
	 \mathrm{s. t.} \ \
	&\displaystyle  \mathrm{Pr}\bigl\{ Q_k(t) \geq Q_{k}^{\mathrm{th}} \bigr\} \leq \epsilon, \ \ \forall k, \ \forall t, \label{eq:P1-C2} \\
	%
	%
	%
	%
	\begin{split}
	\label{eq:P1-C5}
	&\displaystyle R_k(t) \leq R_k^{\mathrm{max}}(t), \ \ \forall k, \ \forall t, \end{split}  \\
	%
	%
	%
    %
	\begin{split}
	\label{eq:P1-C6}
	&\displaystyle E_k^h(t) \leq E_k^{\mathrm{max}}(t), \ \ \forall k, \ \forall t, \end{split}  \\
	& \displaystyle  0 \leq \rho_k(t) \leq 1, \ \ \forall k, \ \forall t,  \label{eq:P1-C4}
\end{align}
\end{subequations}
where the expectation operation in \eqref{eq:P1} is with respect to random channel states and data arrivals processes. 
Note that~\eqref{eq:P1-C2} guarantees that the queue backlog of $k$-th \ac{UE} at each time slot~$t$ is less than $Q_k^{\mathrm{th}}$ with probability $1-\epsilon$, thus, it  ensures the desired {users' latency} requirements.

%
%

\section{Dynamic Control Algorithm via Lyapunov Framework}
\label{sec:Lyapunov}

Problem~\eqref{eq:P1-prob} consists of a time-average sum-power objective function~\eqref{eq:P1}, a non-linear probabilistic constraint~\eqref{eq:P1-C2}, rate constraints~\eqref{eq:P1-C5} including coupled and non-convex \ac{SINR} expressions, \ac{EH} constraint~\eqref{eq:P1-C6} including time-correlated  battery energy dynamics, which cannot be addressed directly in a tractable manner. In this section, we first handle the time-average and probabilistic intractability, and then {provide an online dynamic control algorithm exploiting the drift-plus-penalty framework and the Lyapunov optimization theory}~\cite{neely2010stochastic}. The proposed convex relaxations for the coupled and non-convex constraints are later provided in Section~\ref{sec:Prop_Approximation}.

We start by applying the well-known Markov's inequality~\cite{gross2008fundamentals}, and transform the non-linear probabilistic queue-length constraint~\eqref{eq:P1-C2} as $\mathrm{Pr}\bigl\{ Q_k(t) \geq Q_{k}^{\mathrm{th}} \bigr\} \leq \mathbb{E}[Q_k(t)]\bigl/Q_{k}^{\mathrm{th}} \leq \epsilon,$ $\forall k\in\mathcal{K}$. Thereby, problem~\eqref{eq:P1-prob} can be equivalently rewritten as
\begin{subequations}
\label{eq:P1-prob-X}
\begin{align}
	\displaystyle 
     \underset{\mathbf{f}_k(t), \rho_k(t), \ \forall t}{\mathrm{min}} \ \ & \lim_{T \to \infty} \frac{1}{T}  \sum\limits_{t=0}^{T-1}  \sum\limits_{k \in \mathcal{K}} \mathbb{E}\bigl[\|\mathbf{f}_{k}(t)\|^2\bigr]  \label{eq:P1-X} \\
	 \mathrm{s. t.} \ \
	&\displaystyle  \lim_{T \to \infty} \frac{1}{T}  \sum\limits_{t=0}^{T-1}  \mathbb{E}\bigl[Q_k(t)\bigr] \leq \epsilon Q_{k}^{\mathrm{th}}, \ \ \forall k, \label{eq:P1-C2-X} \\
	&\displaystyle \mathrm{constraints} \ \ \eqref{eq:P1-C5}  -  \eqref{eq:P1-C4}. \nonumber
\end{align}
\end{subequations}

Next, the time-average constraint~\eqref{eq:P1-C2-X} is tackled by recasting it as a queue stability problem~\cite[Ch.5]{neely2010stochastic}. Specifically, {we define a virtual data queue~$Z_k(t)$ for each user~$k$}, which evolves~as 
\begin{equation}
    \label{eq:Virtual-Queue}
    Z_k(t+1) = \bigl[Z_k(t) + Q_k(t+1) - \epsilon Q_k^{\mathrm{th}}\bigr]^{+},  \ \ \forall k.
\end{equation}
Note that the inequality constraint~\eqref{eq:P1-C2-X} is ensured only if the associated virtual queues $\{Z_k(t)\}_{\forall k}$ are stabilized \cite[Theorem 2.5]{neely2010stochastic}. Thus, to stabilize the virtual queues in~\eqref{eq:Virtual-Queue}, we employ the Lyapunov optimization framework, and design a dynamic control algorithm with  long-term stability requirements. It is worth highlighting that the feasible control action sets are coupled over time due to the temporally correlated {user-specific battery charge level~\eqref{eq:P1-C6}}~\cite{Zhaoxia-BatteryConf-2016, Sun-BatteryConf-2018}. 
To circumvent this issue, we introduce a virtual battery queue~$U_k(t)$ for each user~$k$, defined~as
\begin{equation}
    \label{eq:Pertubed_Battery}
    U_k(t) \triangleq B_k^{\mathrm{max}}-B_k(t),  \ \ \forall k.
\end{equation}
%
%
From~\eqref{eq:Pertubed_Battery}, we can observe that the spare energy availability at the receivers' battery increases inversely with the value of $U_k(t), \ \forall k\!\in\!\mathcal{K}$. 
%
Thus, to jointly stabilize the  coupled virtual data queues~\eqref{eq:Virtual-Queue} and  battery queues~\eqref{eq:Pertubed_Battery}, we now define a weighted quadratic Lyapunov function~\cite{neely2010stochastic} 
\begin{equation}
    \label{eq:Lyapunov-def}
    \mathcal{L}({\mathbf \Psi}(t)) \triangleq \frac{1}{2} \sum\limits_{k\in\mathcal{K}} \bigl\{Z_k(t)^{2} + {\omega_k}U_k(t)^{2} \bigr\},
\end{equation}
where ${\mathbf \Psi}(t) \! = \! \bigl[Z_k(t), Q_k(t), U_k(t), B_k(t) \bigl| \forall k \! \in \! \mathcal{K}\bigr]$ and $\omega_k\!>\!0$ is a linear scaling factor on the spare battery capacity of the $k$-th \ac{UE}. 
Note that the parameter $\{\omega_k\}_{\forall k}$ in~\eqref{eq:Lyapunov-def} can be considered as a free parameter (e.g., it may reflect energy to data bits conversion factor), which balances the user-specific contributions of latency-battery queue in the Lyapunov function. Specifically, different scaling on the virtual battery queue might lead to a different Pareto optimal allocation of  resources\footnote{The user-specific priority weights~$\{\omega_k\}_{\forall k}$ corresponding to the spare battery capacity~\eqref{eq:Lyapunov-def} are design parameters that can be tuned based on available statistical information, e.g., latency-battery queue states and channel conditions, to achieve the desired trade-off between BS transmit energy and receiver battery charge level. In fact, this is an interesting topic for future extensions.}. 
%
%
%
Then, the Lyapunov drift between two consecutive time slot is given~by
\begin{subequations}
\label{eq:Lyapunov-drif-def}
\begin{align}
\label{eq:Lyapunov-drif-1}
    \triangle ({\mathbf \Psi}(t)) & =  \frac{1}{2} \mathbb{E} \bigl[  \mathcal{L}({\mathbf \Psi}(t+1)) - \mathcal{L}({\mathbf \Psi}(t))  \bigl|  {\mathbf \Psi}(t) \bigr] \\ \displaybreak[0]
\label{eq:Lyapunov-drif-2}
    & =  \frac{1}{2} \mathbb{E} \biggl[ \sum\limits_{k\in\mathcal{K}} \Big\{ \bigl(Z_k(t+1)^{2} - Z_k(t)^{2} \bigr) + {\omega_k}\bigl(U_k(t+1)^{2} - U_k(t)^{2} \bigr) \Big\} \Bigl|  {\mathbf \Psi}(t) \biggr].
\end{align}
\end{subequations}

Next, by using expressions~\eqref{eq:Queue_Dynamics}, \eqref{eq:Battery_Dynamics}, \eqref{eq:Virtual-Queue} and~\eqref{eq:Pertubed_Battery} in \eqref{eq:Lyapunov-drif-def}, and after some algebraic simplifications, an upper bound for the Lyapunov drift~$\triangle({\mathbf \Psi}(t))$ is obtained~as
\begin{align}
    \label{eq:Lyapunov-drif}
    \triangle ({\mathbf \Psi}(t)) & \leq  \Omega  - \mathbb{E}\Bigl[   \sum\limits_{k\in\mathcal{K}} \bigl(Q_k(t) + A_k(t) + Z_k(t)\bigr) R_k(t) \bigl| {\mathbf \Psi}(t) \Bigr] \nonumber \\ & \qquad \qquad \qquad - \mathbb{E} \Bigl[   \sum\limits_{k\in\mathcal{K}} \omega_k \bigl(B_k^{\mathrm{max}} -  B_k(t)\bigr) E_k^h(t) \bigl|  {\mathbf \Psi}(t) \Bigr], 
\end{align}
where $\Omega$ is a positive constant term. 
More specifically, $\Omega=\zeta_Q+ \Phi_Q(t) + \zeta_B+ \Phi_B(t)$, while the following conditions\footnote{To obtain~\eqref{eq:Lyapunov-drif}, we used $([a_1\!+\!a_2\!-\!a_3]^{+})^2 \leq (a_1\!+\!a_2\!-\!a_3)^2, \ \forall \{a_1,a_2,a_3\}\geq 0$. 
{Further, we have assumed that the second moments of data arrival  $\mathbb{E} [A_k(t)^2 | {\mathbf \Psi}(t) ]$ and transmission $\mathbb{E} [R_k(t)^2 | {\mathbf \Psi}(t) ]$ processes in~\eqref{eq:zeta_Q}, and energy utilized~$\mathbb{E} [E_k^u(t)^2 | {\mathbf \Psi}(t) ]$ and harvested~$\mathbb{E} [E_k^h(t)^2 | {\mathbf \Psi}(t) ]$ processes in~\eqref{eq:zeta_B} for all $k\in\mathcal{K}$ are bounded~\cite[Ch.~5]{neely2010stochastic}.}} are satisfied in all time~slots~\cite{neely2010stochastic}:
\begin{subequations}
\begin{align}
\zeta_Q & \geq \frac{1}{2} \mathbb{E} \Bigl[ \sum\limits_{k\in\mathcal{K}} \bigl\{ A_k(t)^2 +  {R}_k(t)^2 \bigr\}  \bigl| {\mathbf \Psi}(t) \Bigr],  \label{eq:zeta_Q} \\
\Phi_Q(t) & =  \sum\limits_{k\in\mathcal{K}} \Bigl\{ \frac{1}{2} (\epsilon Q_k^{\mathrm{th}})^2 + \frac{1}{2} Q_k(t)^2 + Z_k(t)Q_k(t) + \bigl(Q_k(t) + Z_k(t)\bigr){A}_k(t) \Bigr\} , \\
\zeta_B & \geq \frac{1}{2} \mathbb{E} \Bigl[ \sum\limits_{k\in\mathcal{K}} \bigl\{ \omega_k E_k^u(t)^2 + \omega_k E_k^h(t)^2 \bigr\} \bigl| {\mathbf \Psi}(t) \Bigr],  \label{eq:zeta_B} \\
\Phi_B(t) & =  \sum\limits_{k\in\mathcal{K}} {\omega_k} \bigl(B_k^{\mathrm{max}} -  B_k(t)\bigr)E_k^u(t) .
\end{align}
\end{subequations}

We now define the {drift-plus-penalty} function~\cite{neely2010stochastic} for~\eqref{eq:P1-prob-X} as
\begin{align}
    \label{eq:drift-penality}
     &\triangle ({\mathbf \Psi}(t)) + V \mathbb{E} \Bigl[\sum\limits_{k \in \mathcal{K}} \|\mathbf{f}_{k}(t)\|^2  \bigl| {\mathbf \Psi}(t) \Bigr], 
\end{align}
{where $V\!\geq\!0$ is a trade-off parameter.} 
Thus, by minimizing the upper bound of~\eqref{eq:drift-penality} subject to constraints~\eqref{eq:P1-C5}~\!-~\!\eqref{eq:P1-C4} at each time slot, we can ensure the user-specific latency requirements and maintain the required battery energy levels at the receiver, while minimizing the  sum-power objective function~\eqref{eq:P1}. Herein, we employ the concept of {opportunistic minimization of an expectation}~\cite[Ch. 1.8]{neely2010stochastic} to minimize the upper bound of~\eqref{eq:drift-penality}, and {obtain a series of per-time-slot deterministic problem. Moreover, each subproblem~(21) can be independently solved based on the current queue buffer, battery charge level, and CSI, as summarized in Algorithm~1.}
%
%
\SetArgSty{textnormal}
\begin{algorithm}[]
	\caption{Dynamic control algorithm for~\eqref{eq:P1-prob}} 
	\label{algLyapunov}
	\SetAlgoLined
	For a given time slot~$t$, obtain battery levels~$\{B_k(t)\}$ and compute $R_k^{\mathrm{max}}(t)$ using~\eqref{eq:Max_Data_Rate} \\
	Observe the current queue backlogs $\{{Q}_k(t), {Z}_k(t)\}$, and solve the following problem:
\begin{subequations}
\label{eq:P1b}
\begin{align}
	\displaystyle 
	\begin{split}
  \underset{\mathbf{f}_k(t), {\gamma}_k(t), e_k(t), \rho_k(t)}{\mathrm{min}} \ \ & V  \sum\limits_{k \in \mathcal{K}} \|\mathbf{f}_{k}(t)\|^2 \! - \!  \sum\limits_{k\in\mathcal{K}} \! \bigl( Q_k(t) + A_k(t) + Z_k(t) \bigr)\log_{2}(1+\gamma_k(t))  \\ & \qquad \qquad \qquad \qquad \qquad \quad - \! \sum\limits_{k\in\mathcal{K}} \omega_k \bigl(B_k^{\mathrm{max}} -  B_k(t)\bigr) e_k(t)
	\end{split} \label{eq:P1b-objective} \\
	 \mathrm{s. t.} \ \
     \begin{split}
	 &\displaystyle  \gamma_k(t) \leq \frac{\rho_k(t)|\mathbf{h}_{k}\herm(t) \mathbf{f}_{k}(t)|^2}{\rho_k(t) \! \sum\limits_{u \in \mathcal{K} \backslash k} \! \! |\mathbf{h}_{k}\herm(t) \mathbf{f}_{u}(t)|^2 \! + \! \rho_k(t)\sigma_k^2 \!+ \! {\delta_k^2}}, \ \ \forall k, \end{split} \label{eq:P1b-C2} \\
	\begin{split}
	&\displaystyle {e}_k(t) \leq \zeta_k\bigl(1-\rho_k(t)\bigr)\Bigl(\sum\limits_{j=1}^{K} |\mathbf{h}_{k}\herm(t) \mathbf{f}_{j}(t)|^2  \! + \! \sigma_k^2 \Bigr),
	\ \ \forall k, 	\end{split} \label{eq:P1b-C3} \\
	\begin{split}
	\label{eq:P1b-C5}
	&\displaystyle \log_2(1+\gamma_k(t)) \leq R_k^{\mathrm{max}}(t), \ \ \forall k, \end{split}  \\
	\begin{split}
	\label{eq:P1b-C6}
	&\displaystyle {e}_k(t) \leq E_k^{\mathrm{max}}(t), \ \ \forall k, \end{split}  \\
	& \displaystyle  0 \leq \rho_k(t) \leq 1, \ \ \forall k.  \label{eq:P1b-C4}
\end{align}
\end{subequations} 
    \\
    Update queues $\{{Q}_k(t+1)\}$ and $\{{Z}_k(t+1)\}$ by using~\eqref{eq:Queue_Dynamics} and~\eqref{eq:Virtual-Queue}, respectively \\
    Update battery levels $\{{B}_k(t+1)\}$ by using~\eqref{eq:Battery_Dynamics} \\
    Set $t = t +1$, and go to step~$1$
\end{algorithm}
{Note that we have conveniently introduced the non-negative auxiliary variables~$\{\gamma_k(t), e_k(t)\}_{\forall k}$ to ease the problem tractability. Nevertheless, this does not affect the corresponding solution since
constraints~\eqref{eq:P1b-C2} and~\eqref{eq:P1b-C3} hold with equality at the optimality~\cite{boyd2004convex}. 
Moreover}, for a given queue state at beginning of each time slot~$t$, the transmit beamformers and receive \ac{PS} ratios can be computed independently as described in   Section~\ref{sec:Prop_Approximation}.

%
%

\vspace*{-5pt}
\section{Joint Beamformers and PS Ratios Design}
\label{sec:Prop_Approximation}

The joint beamformers and PS ratios design problem~\eqref{eq:P1b} introduced in Algorithm~\ref{algLyapunov} is still intractable due to the non-convex \ac{SINR} expressions~\eqref{eq:P1b-C2}, and the coupling between the optimization variables~$\{\mathbf{f}_k(t), \rho_k(t), \gamma_k(t)\}$ in~\eqref{eq:P1b-C2} and the optimization variables~$\{\mathbf{f}_k(t), \rho_k(t), e_k(t)\}$ in~\eqref{eq:P1b-C3}.
In this section, {the coupled and non-convex constraints in~\eqref{eq:P1b} 
are handled by using two alternative approaches}: i) \ac{SDR} combined with recently proposed \ac{FP} {quadratic transform} technique~\cite{FP_Shen_2018}, and ii) linear Taylor series approximation via \ac{SCA} framework~\cite{beck2010sequential}. 
Both proposed methods achieve efficient solutions, each with specific convergence  and  complexity characteristics. Specifically, for the considered setup configurations, the \ac{SDR}-FP based method provides faster convergence in terms of the required approximation point updates. 
In contrast, the \ac{SCA} framework has  considerably lower computational complexity per iteration and the optimized beamforming vectors are obtained directly, which will become clear in the following.

It is worth highlighting that due to the underlying interdependence of latency-battery queue dynamics, the  constrained optimization schemes proposed in~\cite{Qingjiang-TWC-SWIPTKey_2014, Shi-SWIPT-2014, Zhang-TVT-2017, Dong-ICC-2016, Baosheng-Zhang-ComLet-2016, Timotheou-TWC-2014, Zong-TWC-2016, Qingjiang-TSP-2016} can not be used directly to solve~\eqref{eq:P1b}. Thus, our proposed solutions are significantly more advanced, and provide a {systematic} approach for handling {problems with} mutually coupled constraints. 
Further, proposed approaches provide valuable insights on the joint optimization of beamformers and PS ratios that concurrently satisfies the latency and EH requirements per user,  
which is inherently an intractable and NP-hard optimization problem.
%
%
Note that in both approaches, the non-convex problem~\eqref{eq:P1b} is first recast as a sequence of convex subproblems, and then iteratively solved until the desired convergence of the objective function~\eqref{eq:P1b-objective}. In the following, we omit the time index~$t$ to simplify the notation.


\vspace*{-10pt}
\subsection{Solution {of~\eqref{eq:P1b} via SDR-FP} framework}
\label{subsec:SDR}
\vspace*{-5pt}

To begin with, let us define $\mathbf{F}_k \! \triangleq \! \mathbf{f}_k\mathbf{f}_k\herm,$ $\forall k\in\mathcal{K}$. {The main idea behind the \ac{SDR} framework is to {drop the rank-one constraints}, e.g., replace each semidefinite rank-one matrix $\mathbf{F}_k$ by a general-rank positive semidefinite matrix, i.e., $\mathbf{F}_k \! \succeq  \! 0, \ \forall k \! \in \! \mathcal{K}$. Then, an additional post-approximation procedure is required to obtain rank-one  beamforming vectors using, e.g., eigen-decomposition and possibly combined with some randomization techniques \cite{Luo-SDR-2010, Harri_2014}, as will become clear in the following.} Thus, by applying the standard SDR technique to~\eqref{eq:P1b}, and after some algebraic simplifications, we obtain the following {relaxed} problem
\begin{subequations}
\label{eq:P1-prob-2}
\begin{align}
	\displaystyle
	\begin{split}
	 \underset{\mathbf{F}_k, {\gamma}_k,  e_k, \rho_k}{\mathrm{min}} \ \ & V \sum\limits_{k \in \mathcal{K}} \mathrm{Tr}(\mathbf{F}_k) -  \sum\limits_{k\in\mathcal{K}} ( Q_k +  A_k  + Z_k )\log_{2}(1+\gamma_k) -   \sum\limits_{k\in\mathcal{K}} \omega_k (B_k^{\mathrm{max}} -  B_k) e_k
	\end{split}
	 \label{eq:P1-2} \\
	 \mathrm{s. t.} \ \
	\begin{split}
	& \displaystyle \frac{1}{\gamma_k} \mathbf{h}_k\herm \mathbf{F}_k \mathbf{h}_k  - \!\! \sum\limits_{u \in \mathcal{K} \backslash k} \!\! \mathbf{h}_{k}\herm \mathbf{F}_u \mathbf{h}_{k} \geq \sigma_k^2 +  \frac{\delta_k^2}{\rho_k},  \ \ \forall k, 	
	\end{split} \label{eq:P1-C1-2} \\
	\begin{split}
	& \displaystyle \frac{1}{e_k} \Bigl( { {\sum\limits_{j=1}^{K}} \mathbf{h}_k\herm \mathbf{F}_j \mathbf{h}_k} + \sigma_k^2 \Bigr) \geq \frac{1}{\zeta_k(1-\rho_k)} , \ \ \forall k, 	
	\end{split} \label{eq:P1-C2-2} \\
	\begin{split}
	\label{eq:P1b-C5-2}
	&\displaystyle \gamma_k \leq 2^{R_k^{\mathrm{max}}}-1, \ \ \forall k, \end{split}  \\
	\begin{split}
	\label{eq:P1b-C6-2}
	&\displaystyle {e}_k \leq E_k^{\mathrm{max}}, \ \ \forall k, \end{split}  \\
	& \displaystyle  0 \leq \rho_k \leq 1, \ \ \forall k,  \label{eq:P1-C4-2} \\
	& \displaystyle  {\mathbf{F}_k \succeq 0, \ \  \forall k},  \label{eq:P1-C3-2} 
\end{align}
\end{subequations}
where $\mathrm{Tr}(\cdot)$ denotes the trace operation. Note that the reformulated problem~\eqref{eq:P1-prob-2} is still non-convex due to the coupling between variables $\{\mathbf{F}_k , \gamma_k\}$ in~\eqref{eq:P1-C1-2} and $\{\mathbf{F}_k , e_k\}$ in~\eqref{eq:P1-C2-2}. 
Thus, managing the mutually coupled EH and SINR combinations is considerably more challenging than conventional constrained optimization with a fixed QoS requirement per user~\cite{Qingjiang-TWC-SWIPTKey_2014, Shi-SWIPT-2014}. 
%
%
To handle the fractional non-convexity in~\eqref{eq:P1-prob-2}, we use the recently proposed \ac{FP} {quadratic transform} technique~\cite{FP_Shen_2018}, also reproduced in Proposition~\ref{Prop1}, that has been shown to converge to a local solution.
Specifically, the \ac{FP} quadratic transform techniques decouple the numerator and the denominator of the concave-convex functions, thereby enabling the iterative optimization between primal  and auxiliary variables. 
This strategy works well for a variety of optimization problems, including scheduling, power control, and beamformer design~\cite{FP_Shen_2018}.

\begin{proposition}\label{Prop1}
Given a non-negative function~$A(\mathbf{x})\geq0$, a strictly positive function~$B(\mathbf{x})>0$, and a monotonic function $f(\cdot): \mathbb{R}\mapsto\mathbb{R}$, the coupled fractional constraint
%
%
%
%
%
%
$f\Bigl(\frac{A(\mathbf{x})}{B(\mathbf{x})}\Bigr)$ is equivalent to $f\bigl(2 \varpi \sqrt{A(\mathbf{x})} - \varpi_k^2 B(\mathbf{x}) \bigr)$, when the auxiliary variable $\varpi$ has the optimal value $\varpi^{\scriptsize (\star)}\! = \! \frac{\sqrt{A(\mathbf{x})}}{B(\mathbf{x})}$.
\end{proposition}
\begin{proof}
Proposition~\ref{Prop1} can be easily proved by following the steps in~\cite[Section II]{FP_Shen_2018}.
\end{proof}
%
%
%
%

%
%
%
\subsubsection{Convex approximation for constraint~\eqref{eq:P1-C1-2}}
\label{subsec:SINR_Approximation_SDR}~\\
\noindent By using Proposition~\ref{Prop1}, the  fractional term in~\eqref{eq:P1-C1-2}
\begin{subequations}
\begin{equation}
\label{eq:Fractional term 1}
\frac{1}{\gamma_k} \mathbf{h}_k\herm \mathbf{F}_k \mathbf{h}_k,  \ \ \forall k,
\end{equation}
is equivalent to
\begin{equation}
\label{eq:Fractional term 2}
2 \nu_k \sqrt{\mathbf{h}_k\herm \mathbf{F}_k \mathbf{h}_k} - \nu_k^2 \gamma_k,  \ \ \forall k,
\end{equation}
when the auxiliary variable $\nu_k$ has the optimal value
\begin{equation}
\label{eq:Fractional term 3}
\nu_k^{\scriptsize (\star)} =  \frac{1}{\gamma_k} \sqrt{\mathbf{h}_k\herm \mathbf{F}^{}_k \mathbf{h}_k},  \ \ \forall k.
\end{equation}
\end{subequations}
%
%
%
%
%
%
%

\subsubsection{Convex approximation for constraint~\eqref{eq:P1-C2-2}}
\label{subsec:EH_Approximation_SDR}~\\
\noindent Similarly, by using Proposition~\ref{Prop1}, the  fractional term in~\eqref{eq:P1-C2-2}
\begin{subequations}
\begin{equation}
\label{eq:Fractional term 1b}
\frac{1}{e_k} \Bigl( { {\sum\limits_{j=1}^{K}} \mathbf{h}_k\herm \mathbf{F}_j \mathbf{h}_k} + \sigma_k^2 \Bigr),  \ \ \forall k,
\end{equation}
is equivalent to
\begin{equation}
\label{eq:Fractional term 2b}
2 \mu_k \sqrt{  {\textstyle{\sum_{j=1}^{K}}} \mathbf{h}_k\herm \mathbf{F}_j \mathbf{h}_k + \sigma_k^2 } - \mu_k^2 e_k,  \ \ \forall k,
\end{equation}
when the associated auxiliary variable $\mu_k$ has the optimal value
\begin{equation}
\label{eq:Fractional term 3b}
\mu_k^{\scriptsize (\star)} =  \frac{1}{e_k} \sqrt{{\textstyle{\sum_{j=1}^{K}}} \mathbf{h}_k\herm \mathbf{F}_j \mathbf{h}_k + \sigma_k^2},  \ \ \forall k.
\end{equation}

\end{subequations}


\subsubsection{Approximated convex subproblem}~\\
\noindent {By setting a fixed approximation point $\{\nu\iLoop_k, \mu\iLoop_k\}$, and substituting~\eqref{eq:Fractional term 2} and \eqref{eq:Fractional term 2b} into \eqref{eq:P1-C1-2} and \eqref{eq:P1-C2-2}, respectively, the problem~\eqref{eq:P1-prob-2} can be expressed in convex form as}
\vspace*{-4pt}
\begin{subequations}
\label{eq:P1-prob-3}
\begin{align}
	\displaystyle
	\begin{split}
	 \underset{\mathbf{F}_k, {\gamma}_k, e_k, \rho_k}{\mathrm{min}} \ \ &  V \sum\limits_{k \in \mathcal{K}} \mathrm{Tr}(\mathbf{F}_k)  
	 - \sum\limits_{k\in\mathcal{K}} ( Q_k + A_k + Z_k )\log_{2}(1+\gamma_k) -  \sum\limits_{k\in\mathcal{K}}  \omega_k(B_k^{\mathrm{max}} -  B_k) e_k
	\end{split}
	 \label{eq:P1-3} \\
	 \mathrm{s. t.} \quad
	\begin{split}
	& \displaystyle 2 \nu\iLoop_k \sqrt{\mathbf{h}_k\herm \mathbf{F}_k \mathbf{h}_k} - \bigl(\nu\iLoop_k\bigr)^2 \gamma_k 
	-  \sum\limits_{u \in \mathcal{K} \backslash k} \!\! \mathbf{h}_k\herm \mathbf{F}_u \mathbf{h}_k \geq  \sigma_k^2 + \frac{\delta_k^2}{\rho_k},  \ \ \forall k, 	
	\end{split} \label{eq:P1-C1-3} \\
	\begin{split}
	& \displaystyle 2 \mu\iLoop_k \sqrt{  { \sum\nolimits_{j=1}^{K}} \mathbf{h}_k\herm \mathbf{F}_j \mathbf{h}_k + \sigma_k^2 } - \bigl(\mu\iLoop_k\bigr)^2 e_k \geq \frac{1}{\zeta_k(1-\rho_k)}, \ \ \forall k, 	
	\end{split} \label{eq:P1-C2-3} \\
	\begin{split}
	\label{eq:P1-C5-3}
	&\displaystyle \gamma_k \leq 2^{R_k^{\mathrm{max}}}-1, \ \ \forall k, \end{split}  \\
	\begin{split}
	\label{eq:P1b-C6-3}
	&\displaystyle {e}_k \leq E_k^{\mathrm{max}}, \ \ \forall k, \end{split}  \\
	& \displaystyle  0 \leq \rho_k \leq 1, \ \ \forall k,  \label{eq:P1-C4-3} \\
	& \displaystyle  \mathbf{F}_k 	\succeq 0, \ \ \forall k.  \label{eq:P1-C3-3} 
\end{align}
\end{subequations}
Note that~\eqref{eq:P1-prob-3} provides an approximate solution for~\eqref{eq:P1-prob-2} in the vicinity of $\{\nu\iLoop_k, \mu\iLoop_k\}$. Thus, by iteratively solving~\eqref{eq:P1-prob-3} while updating~$\{\nu\iLoop_k, \mu\iLoop_k\}$ with the solution of current iteration, we can find a  solution for~\eqref{eq:P1-prob-2}~\cite{FP_Shen_2018}.
Moreover, for a fixed operating point~$\{\nu\iLoop_k, \mu\iLoop_k\}$, the subproblem~\eqref{eq:P1-prob-3} can be efficiently solved using convex optimization tools, e.g., $\mathrm{CVX}$~\cite{cvx}.

For a given time slot~$t$, let $\big\{\mathbf{F}^{\star}_k(t), \gamma^{\star}_k(t), e^{\star}_k(t), \rho^{\star}_k(t) \big\}_{\forall k}$ denote the  solution obtained from~\eqref{eq:P1-prob-3}. However, note that the achievable \ac{SINR}~\eqref{eq:SINR-ID} and the maximum harvested power~\eqref{eq:power-EH} at $k$-th UE is computed by using actual beamforming vectors~$\mathbf{f}^{\star}_k(t)$ and \ac{PS} ratios~$\rho^{\star}_k(t)$. 
Therefore, {if $\mathbf{F}^{\star}_k(t)$ satisfies $\mathrm{Rank}\bigl(\mathbf{F}^{\star}_k(t)\bigr)\!=\!1, \forall k\in\mathcal{K}$, we can  write $\mathbf{F}^{\star}_k(t)\!=\!\mathbf{f}^{\star}_k(t) \mathbf{f}_k^{\star\mbox{\scriptsize H}}(t)$, and $\mathbf{f}^{\star}_k(t)$ is a feasible solution to original problem~\eqref{eq:P1b}. On the contrary, if $\mathrm{Rank}\bigl(\mathbf{F}^{\star}_k(t)\bigr)\!>\!1$, an additional post-approximation procedure is required to obtain the beamforming vectors, which inevitably degrades the achievable solution~\cite{Harri_2014, Luo-SDR-2010}. For example, one possible rank-one approximation 
is obtained by simply setting $\{\mathbf{f}_k^{\star}(t)\}_{\forall k}$ to be proportional to the eigenvector $\mathbf{q}_k^1(t)$ of $\{\mathbf{F}_k^{\star}(t)\}_{\forall k}$ associated with the largest eigenvalue~$\lambda^1_k(t)$, i.e., $\mathbf{f}^{\star}_k(t)\!=\!\sqrt{\varrho^{\star}_k(t)}
\mathbf{q}_k^1(t), \forall k\!\in\!\mathcal{K},$ 
where $\varrho_k$ is scaled to satisfy all the constraints in~\eqref{eq:P1-prob-3}~\cite{Harri_2014}. Specifically, after $\{\mathbf{q}_k^1(t)\}_{\forall k}$ are obtained,~\eqref{eq:P1-prob-3} is modified by replacing $\mathbf{F}_k(t)$ with $\widetilde{\mathbf{F}}_k(t)$, and introducing additional constraints $\widetilde{\mathbf{F}}_k(t)\!=\!\varrho_k(t)\mathbf{q}_k^1(t) \mathbf{q}_k^{1\mbox{\scriptsize H}}(t)$, and then solve the modified problem with optimization variables~${\varrho_k}, \forall k\!\in\!\mathcal{K}$  
%
%
(we refer the reader to~\cite{Harri_2014, Luo-SDR-2010} for more details).} 
%
%
The joint optimization of beamformers and \ac{PS} ratios, with the proposed \ac{SDR}-\ac{FP} based convex approximation, has been summarized in Algorithm~\ref{algFP}.
\vspace*{-8pt}
\SetArgSty{textnormal}
\begin{algorithm}[]
\DontPrintSemicolon
	\caption{{SDR-FP based} iterative algorithm for~\eqref{eq:P1-prob-3}}
	\label{algFP}
	\SetAlgoLined
	Initialize 	$\{\nu_k\oLoop, \mu_k\oLoop\}, \ \forall k$, with feasible values, and set $i \!=\! 1$ \\
	\Repeat{convergence or for fixed number of iterations}{
	{Solve \eqref{eq:P1-prob-3} with $\{\nu_k\xiLoop, \mu_k\xiLoop\}$, and denote the local solution as $\bigl\{\mathbf{F}_{k}\iLoop, \gamma_k\iLoop, e_k\iLoop, {\rho}_k\iLoop\bigr\}$ \\
	Update $\nu_{k}\iLoop$ using~\eqref{eq:Fractional term 3} with $\bigl\{\mathbf{F}_{k}\iLoop, {\gamma}_k\iLoop\bigr\}$, and $\mu_{k}\iLoop$ using~\eqref{eq:Fractional term 3b} with $\bigl\{\mathbf{F}_{k}\iLoop, e_k\iLoop\bigr\}$  \\
	Set $i = i+1$
	}
}
\end{algorithm}
\vspace*{-8pt}
%


%
\subsection{Solution {of~\eqref{eq:P1b} via SCA} framework}
\label{subsec:SCA}
\vspace*{-8pt}

In this subsection, we employ the \ac{SCA} framework~\cite{beck2010sequential}, wherein the  constraints~\eqref{eq:P1b-C2} and~\eqref{eq:P1b-C3} are successively upper-bounded with a sequence of convex subsets via first-order linear Taylor series approximations. The underlying convex subproblem is then iteratively solved until the  desired convergence of the objective function. 
Note that, in contrast to the above-mentioned \ac{SDR}-\ac{FP} technique, here the optimized beamformers are obtained directly from the feasible solution, and no additional post-processing steps are required.
%
%
%
%
%
%
%
The \ac{SCA} based solutions have been widely used in many practical applications, e.g., spectrum sharing, energy efficiency, and multi-antenna interference coordination. For example, the SCA based linear Taylor series relaxation of the concave-convex fractional constraints is provided in our earlier work~\cite{dileep_ISWCS, Dileep_Globecom2020}. 
In the following, the main steps are briefly reproduced.


\subsubsection{Convex approximation for constraint~\eqref{eq:P1b-C2}}
\label{subsec:SINR_Approximation}~\\
\noindent We start by rewriting \ac{SINR} constraint~\eqref{eq:P1b-C2} as 
\begin{equation}
    \label{eq:SINR-rho-Divide}
    \gamma_k \leq \frac{|\mathbf{h}_{k}\herm \mathbf{f}_{k}|^2}{ \sum\limits_{u \in \mathcal{K} \backslash k} \! |\mathbf{h}_{k}\herm \mathbf{f}_{u}|^2 + \sigma_k^2 + \frac{\delta_k^2}{\rho_k}}, \ \ \forall k.
\end{equation}
For compact representation, we define new functions~as  
\begin{subequations}
    \label{eq:G_I_function}
\begin{align}
    \label{eq:G_function}
    & {G}_k(\mathbf{f}_k, \gamma_k) \triangleq \frac{|\mathbf{h}_{k}\herm \mathbf{f}_{k}|^2}{\gamma_k}, \\
    \label{eq:I_function}
    & I_k(\mathbf{F}, {\rho_k}) \triangleq \sum\limits_{u \in \mathcal{K} \backslash k} \! |\mathbf{h}_{k}\herm \mathbf{f}_{u}|^2 + \sigma_k^2 + \frac{\delta_k^2}{\rho_k},
\end{align}
\end{subequations}
where ${\mathbf{F}} \triangleq [\mathbf{f}_1, \mathbf{f}_2, \ldots, \mathbf{f}_K]$. Hence, expression~\eqref{eq:SINR-rho-Divide} can be equivalently written~as
\begin{equation}
    \label{eq:SINR_DC}
    I_k({\mathbf{F}}, {\rho_k}) - {G}_k(\mathbf{f}_k, \gamma_k) \leq 0, \ \ \forall k.
\end{equation}
Note that~\eqref{eq:G_function} is a quadratic-over-linear function, and~\eqref{eq:I_function} is a convex function with respect to the optimization variables, hence, the \ac{LHS} of~\eqref{eq:SINR_DC} is a difference of convex functions~\cite[Ch. 3]{boyd2004convex}. Thus, the linear convex approximation of the equivalent constraint~\eqref{eq:SINR_DC} can be obtained by replacing~${G}_k(\mathbf{f}_k, \gamma_k)$ with its first-order Taylor series approximation around a fixed operating point $\{ \mathbf{f}_k\iLoop, \gamma_k\iLoop\}$
\begin{align}
    \label{eq:Lin-Approx_G_function}
     & \widetilde{{G}}_k(\mathbf{f}_k, \gamma_k,  \mathbf{f}_k\iLoop, \gamma_k\iLoop) \triangleq  2 \Re \biggl\{ \frac{\mathbf{f}_k\iherm \mathbf{h}_k \mathbf{h}_k\herm }{\gamma_k\iLoop} \bigl(\mathbf{f}_k - \mathbf{f}_k \iLoop \bigr) \biggr\} 
    + \frac{|\mathbf{h}_{k}\herm \mathbf{f}_{k}\iLoop|^2}{ \gamma_k\iLoop } \biggl( 1 - \frac{\gamma_k-\gamma_k\iLoop}{\gamma_k\iLoop} \biggr).
\end{align}
%
%


\subsubsection{Convex approximation for constraint~\eqref{eq:P1b-C3}}
\label{subsec:EH_approximation}~\\
\noindent To begin with, we rewrite expression~\eqref{eq:P1b-C3} as  
\begin{equation}
    \label{eq:EH_divide}
    \frac{{e}_k}{\zeta_k(1-\rho_k)} \leq  \sum\limits_{j=1}^{K} |\mathbf{h}_{k}\herm \mathbf{f}_{j}|^2  + \sigma_k^2, \ \ \forall k.
\end{equation}
For compact representation, we define new functions as  
\begin{subequations}
    \label{eq:C_S_function}
\begin{align}
    \label{eq:C_function}
    & {C}_k(\rho_k) \triangleq \frac{1}{\zeta_k(1-\rho_k)}, \\
    \label{eq:S_function}
    & S_k(\mathbf{F}, e_k) \triangleq \frac{ \sum_{j=1}^{K} |\mathbf{h}_{k}\herm \mathbf{f}_{j}|^2 + \sigma_k^2 }{e_k}. 
\end{align}
\end{subequations}
Hence, expression~\eqref{eq:EH_divide} can be rewritten as
\begin{equation}
    \label{eq:EH_DC}
    C_k({\rho_k}) - {S}_k(\mathbf{F}, e_k) \leq 0, \ \ \forall k.
\end{equation}
Observe that~\eqref{eq:C_function} and \eqref{eq:S_function} are respectively convex and quadratic-over-linear functions with respect to the optimization variables. Hence, the \ac{LHS} of~\eqref{eq:EH_DC} is again a difference of convex functions~\cite[Ch. 3]{boyd2004convex}. Then, to provide the linear convex approximation of equivalent constraint~\eqref{eq:EH_DC}, we replace~${S}_k(\mathbf{F})$ with its first-order linear Taylor series approximation around a fixed operating point~$\{ \mathbf{F}\iLoop, e_k\iLoop\}$
\begin{align}
    \label{eq:Lin-Approx_S_function}
     & \widetilde{{S}}_k(\mathbf{F}, e_k,  \mathbf{F}\iLoop, e_k\iLoop)  \triangleq  2 \! \sum\limits_{j=1}^{K} \! \Re \Bigl\{\! {\mathbf{f}_j\iherm \mathbf{h}_k \mathbf{h}_k\herm} \bigl(\mathbf{f}_j - \mathbf{f}_j \iLoop \bigr) \! \Bigr\} 
     + \frac{ \sum\limits_{j=1}^{K} |\mathbf{h}_{k}\herm \mathbf{f}\iLoop_{j}|^2  + \sigma_k^2}{ e_k\iLoop } \biggl( 1 - \frac{e_k-e_k\iLoop}{e_k\iLoop} \biggr).
\end{align}

\subsubsection{Approximated convex subproblem}~\\
\noindent After replacing~\eqref{eq:P1b-C2} and~\eqref{eq:P1b-C3} with its linear  Taylor series approximations~\eqref{eq:Lin-Approx_G_function} and~\eqref{eq:Lin-Approx_S_function}, respectively, 
the problem~\eqref{eq:P1b} can be approximated as the following convex subproblem:  
\begin{subequations}
\label{eq:P1c}
\begin{align}
	\displaystyle 
	\begin{split}
     \underset{\mathbf{f}_k, {\gamma}_k, e_k, \rho_k}{\mathrm{min}} \ \ & V   \sum\limits_{k \in \mathcal{K}} \|\mathbf{f}_{k}\|^2   -     \sum\limits_{k\in\mathcal{K}}  (Q_k   + A_k  + Z_k )\log_{2}(1+\gamma_k) -   \sum\limits_{k\in\mathcal{K}}  \omega_k (B_k^{\mathrm{max}} -  B_k) e_k
	\end{split} \label{eq:P1c-objective} \\
	 \mathrm{s. t.} \ \
    \begin{split}
	&\displaystyle  I_k(\mathbf{F}, {\rho_k}) - \widetilde{{G}}_k(\mathbf{f}_k, \gamma_k, \mathbf{f}_k\iLoop, \gamma_k\iLoop) \leq 0, \ \ \forall k, \end{split} \label{eq:P1c-C2} \\
	\begin{split}
	&\displaystyle  C_k({\rho_k}) - \widetilde{{S}}_k(\mathbf{F}, e_k, \mathbf{F}\iLoop, e_k\iLoop) \leq 0, \ \ \forall k, 	\end{split} \label{eq:P1c-C3} \\
	\begin{split}
	\label{eq:P1c-C5}
	&\displaystyle \gamma_k \leq  2^{R_k^{\mathrm{max}}}-1, \ \ \forall k, \end{split}  \\
	\begin{split}
	\label{eq:P1c-C6}
	&\displaystyle {e}_k \leq E_k^{\mathrm{max}}, \ \ \forall k, \end{split}  \\
	& \displaystyle 0 \leq \rho_k \leq 1, \ \ \forall k  \label{eq:P1c-C4} .
\end{align}
\end{subequations} 
Note that the subproblem~\eqref{eq:P1c} provides an approximate solution in the proximity of a fixed operating point. 
Thus, by iteratively solving~\eqref{eq:P1c} with a convex optimization solver, e.g.,~$\mathrm{CVX}$~\cite{cvx}, and updating variables~$\{\mathbf{f}_k\iLoop, \gamma_k\iLoop, e_k\iLoop, \rho_k\iLoop\}$ with the current \ac{SCA} solution, as illustrated in Algorithm~\ref{algSCA}, we obtain a solution for problem~\eqref{eq:P1b}.
%
%
\textbf{\SetArgSty{textnormal}
\begin{algorithm}[]
\DontPrintSemicolon
	\caption{{SCA based} iterative algorithm for~\eqref{eq:P1c}}
	\label{algSCA}
	\SetAlgoLined
	Initialize with feasible starting point $\bigl\{\mathbf{{f}}_{k}\oLoop, {\gamma}_k\oLoop, e_k\oLoop, \rho_k\oLoop \bigr\}, \ \forall k$, and set $i = 1$ \\
	\Repeat{convergence or for fixed number of iterations}{
	{Solve \eqref{eq:P1c} with $\bigl\{\mathbf{f}_{k}\xiLoop\!, {\gamma}_k\xiLoop\!, e_k\xiLoop\!, \rho_k\xiLoop \bigr\}$ and denote the solution as $\bigl\{\mathbf{f}_{k}^{\star}, {\gamma}_k^{\star}, e_k^{\star},  \rho_k^{\star} \bigr\}$ \\
	Update $\bigl \{\mathbf{f}_{k}\iLoop = \mathbf{f}_{k}^{\star}\bigr \} $, $\bigl \{ {\gamma}_k\iLoop = {\gamma}_k^{\star} \bigr \}$, $\bigl \{ {e}_k\iLoop = {e}_k^{\star} \bigr \}$ and $\bigl \{\rho_k\iLoop = \rho_k^{\star}\bigr \} $ \\
	Set $i = i+1$
	}
}
\end{algorithm}}
%


\section{Delay-bounded Batteryless Devices}
\label{subsec:Infinite_Battery_ISWCS}


In this section, we study a special scenario of great practical significance, the case of batteryless \acp{UE}, which 
is key for self-sustainable and battery-free future wireless networks~\cite{Lopez.2021}. Delay bounded batteryless and low-power networks can play a crucial role for e.g., emergency networks and industrial automation scenarios, due to dense deployment and restricted human access.

%
%
As there is no  battery in the user devices, the energy harvested from the RF received signals is immediately available to support the receiver operations in the current time slot~\cite{Qingjiang-TWC-SWIPTKey_2014, Shi-SWIPT-2014}, e.g., based on the \emph{harvest-use} strategy~\cite[Section III]{Meng-HU_HSU-2016}\footnote{The  harvested energy is temporarily stored in a built-in capacitor, which can be immediately used for receiver operations~\cite{Meng-HU_HSU-2016}.}. 
Moreover, the remainder of the extra harvested energy will be discarded by the user devices. Therefore,  constraints~\eqref{eq:P1b-C5} and~\eqref{eq:P1b-C6} are  no longer required for this specific SWIPT-enabled batteryless setup~\cite{Qingjiang-TWC-SWIPTKey_2014}. However, the receiver must be able to harvest {an} adequate amount of energy, i.e., to support its circuit power consumption and decoding operations on the downlink received data. Let {$\ddot{e}_k(t)$ denote a (fixed) minimum harvested power \ac{QoS} requirement of the $k$-th \ac{UE} for seamless receiver operations during time~slot~$t$.} In the following, time index~$t$ is omitted for notation brevity.
Thereby, the problem~\eqref{eq:P1b} can be recast~as 
%
%
\begin{subequations}
\label{eq:P1b_InfB}
\begin{align}
	\displaystyle 
	\begin{split}
    \underset{\mathbf{f}_k, {\gamma}_k, \rho_k}{\mathrm{min}} \ \ & V  \sum\limits_{k \in \mathcal{K}} \|\mathbf{f}_{k}\|^2 -  \sum\limits_{k\in\mathcal{K}} ( Q_k + A_k + Z_k)\log_{2}(1+\gamma_k)  
	\end{split} \label{eq:P1b-objective_InfB} \\
	 \mathrm{s. t.} \ \
     \begin{split}
	 &\displaystyle  \gamma_k  \leq  \frac{\rho_k |\mathbf{h}_{k}\herm \mathbf{f}_{k}|^2}{\rho_k \! \sum\limits_{u \in \mathcal{K} \backslash k} \! |\mathbf{h}_{k}\herm \mathbf{f}_{u}|^2 +  \rho_k\sigma_k^2 +  {\delta_k^2}}, \ \ \forall k, \end{split} \label{eq:P1b-C2_InfB} \\
	\begin{split}
	&\displaystyle \ddot{e}_k \leq \zeta_k\bigl(1-\rho_k\bigr)\Bigl( \sum\limits_{j=1}^{K} |\mathbf{h}_{k}\herm \mathbf{f}_{j}|^2   +  \sigma_k^2 \Bigr),
	\ \ \forall k, 	\end{split} \label{eq:P1b-C3_InfB} \\
	& \displaystyle  0 \leq \rho_k \leq 1, \ \ \forall k.  \label{eq:P1b-C4_InfB}
\end{align}
\end{subequations} 
Problem~\eqref{eq:P1b_InfB} is intractable due to the non-convex \ac{SINR} constraint~\eqref{eq:P1b-C2_InfB}, and the coupling between 
optimization variables~$\{\mathbf{f}_k, \rho_k\}$
in expression~\eqref{eq:P1b-C2_InfB} and~\eqref{eq:P1b-C3_InfB}.

    
We again adopt the \ac{SCA} framework~\cite{beck2010sequential}, wherein the non-convex constraints~\eqref{eq:P1b-C2_InfB} and~\eqref{eq:P1b-C3_InfB} are upper-bounded with their first-order Taylor series approximations.
Note that~\eqref{eq:P1b-C2_InfB} is equivalent to~\eqref{eq:P1b-C2}, and it can be handled as described in Section~\ref{subsec:SINR_Approximation}. Moreover, following the steps presented in Section~\ref{subsec:EH_approximation}, constraint~\eqref{eq:P1b-C3_InfB} can be addressed as follows.

To begin with, we rewrite expression~\eqref{eq:P1b-C3_InfB} as  
\begin{equation}
    \label{eq:EH_divide_InfB}
    \frac{\ddot{e}_k}{\zeta_k\bigl(1-\rho_k\bigr)} \leq  \sum\limits_{j=1}^{K} |\mathbf{h}_{k}\herm \mathbf{f}_{j}|^2  + \sigma_k^2, \ \ \forall k.
\end{equation}
For compact representation, we define new functions as  
\begin{subequations}
    \label{eq:C_S_function_InfB}
\begin{align}
    \label{eq:C_function_InfB}
    & \ddot{C}_k(\rho_k) \triangleq \frac{\ddot{e}_k}{\zeta_k\bigl(1-\rho_k\bigr)}, \\
    \label{eq:S_function_InfB}
    & \ddot{S}_k(\mathbf{F}) \triangleq  \sum\limits_{j=1}^{K} |\mathbf{h}_{k}\herm \mathbf{f}_{j}|^2  + \sigma_k^2. 
\end{align}
\end{subequations}
Hence, expression~\eqref{eq:EH_divide_InfB} can be equivalently rewritten as
\begin{equation}
    \label{eq:EH_DC_InfB}
    \ddot{C}_k({\rho_k}) - \ddot{S}_k(\mathbf{F}) \leq 0, \ \ \forall k.
\end{equation}
We can observe that both~\eqref{eq:C_function_InfB} and~\eqref{eq:S_function_InfB} are convex functions, and hence, the {LHS} of~\eqref{eq:EH_DC_InfB} is a difference of convex functions~\cite[Ch. 3]{boyd2004convex}. 
We replace the quadratic function~$\ddot{S}_k(\mathbf{F})$ 
with its first-order Taylor series approximation around a fixed operating point $\{ \mathbf{F}\iLoop\}$~as
\begin{align}
    \label{eq:Lin-Approx_S_function_InfB}
     & \widetilde{\ddot{S}}_k(\mathbf{F}, \mathbf{F}\iLoop)  \triangleq  2  \sum\limits_{j=1}^{K} \Re \Bigl\{ {\mathbf{f}_j\iherm \mathbf{h}_k \mathbf{h}_k\herm} \bigl(\mathbf{f}_j - \mathbf{f}_j \iLoop \bigr) \Bigr\} 
      + \sum\limits_{j=1}^{K} |\mathbf{h}_{k}\herm \mathbf{f}\iLoop_{j}|^2 + \sigma_k^2 .
\end{align}

Thereby, the problem~\eqref{eq:P1b_InfB} can be approximated as the following convex subproblem around a fixed operating point $\{ \mathbf{f}_k\iLoop, {\gamma}_k\iLoop \}$
\begin{subequations}
\label{eq:P1c_InfB}
\begin{align}
	\displaystyle 
	\begin{split}
     \underset{\mathbf{f}_k, {\gamma}_k, \rho_k}{\mathrm{min}} & \ \ V  \sum\limits_{k \in \mathcal{K}}  \|\mathbf{f}_{k}\|^2 -   \sum\limits_{k\in\mathcal{K}}  (Q_k   + A_k  + Z_k )\log_{2}(1+\gamma_k)
	\end{split} \label{eq:P1c-objective_InfB} \\
	 \mathrm{s. t.} \ \
    \begin{split}
	&\displaystyle \lambda_{k,1} : \ I_k(\mathbf{F}, {\rho_k}) - \widetilde{{G}}_k(\mathbf{f}_k, \gamma_k, \mathbf{f}_k\iLoop, \gamma_k\iLoop) \leq 0, \ \ \forall k, \end{split} \label{eq:P1c-C2_InfB} \\
	\begin{split}
	&\displaystyle \lambda_{k,2} : \ \ddot{C}_k(\rho_k) - \widetilde{\ddot{S}}_k(\mathbf{F}, \mathbf{F}\iLoop) \leq 0, \ \ \forall k, 	\end{split} \label{eq:P1c-C3_InfB} \\
	& \displaystyle \lambda_{k,3} : \ \rho_k \geq 0, \ \ \forall k,  \label{eq:P1c-C4_InfB} \\
	& \displaystyle \lambda_{k,4} : \ \rho_k \leq 1, \ \ \forall k,  \label{eq:P1c-C5_InfB}
\end{align}
\end{subequations} 
where $\bm{\lambda}_k=[\lambda_{k,1},\lambda_{k,2}, \lambda_{k,3}, \lambda_{k,4}]$ are non-negative Lagrange multipliers associated with each constraint. {The role of the Lagrange multipliers will become clear in the following subsection.} Note that  subproblem~\eqref{eq:P1c_InfB} can be iteratively solved using standard convex optimization tools, e.g., $\mathrm{CVX}$~\cite{cvx}, as illustrated in Algorithm~\ref{algSCA_InfB}.
%
%
\textbf{\SetArgSty{textnormal}
\begin{algorithm}[]
\DontPrintSemicolon
	\caption{{SCA based} iterative algorithm for~\eqref{eq:P1c_InfB}}
	\label{algSCA_InfB}
	\SetAlgoLined
	Initialize with feasible starting point $\bigl\{\mathbf{{f}}_{k}\oLoop, {\gamma}_k\oLoop, \rho_k\oLoop \bigr\}, \ \forall k$, and set $i = 1$ \\
	\Repeat{convergence or for fixed number of iterations}{
	{Solve \eqref{eq:P1c_InfB} with $\bigl\{\mathbf{f}_{k}\xiLoop, {\gamma}_k\xiLoop, \rho_k\xiLoop \bigr\}$ and denote the solution as $\bigl\{\mathbf{f}_{k}^{\star}, {\gamma}_k^{\star},  \rho_k^{\star} \bigr\}$ \\
	Update $\bigl \{\mathbf{f}_{k}\iLoop = \mathbf{f}_{k}^{\star}\bigr \} $, $\bigl \{ {\gamma}_k\iLoop = {\gamma}_k^{\star} \bigr \}$ and $\bigl \{\rho_k\iLoop = \rho_k^{\star}\bigr \} $ \\
	Set $i = i+1$
	}
}
\end{algorithm}}
\vspace*{-10pt}
%


\vspace*{-5pt}
\subsection{Solution {of \eqref{eq:P1c_InfB}} via KKT conditions}
\vspace*{-5pt}
Herein, we also provide a low-complexity iterative algorithm that does not rely on generic convex solvers. Specifically, we tackle~\eqref{eq:P1c_InfB} by iteratively solving a system of closed-form \ac{KKT} optimality conditions~\cite[Ch. 5.5]{boyd2004convex}. After some algebraic manipulations, we obtain the Lagrangian $\mathfrak{L}(\mathbf{F}, {\gamma}_k, \rho_k, \bm{\lambda}_k)$ of~\eqref{eq:P1c_InfB} as detailed in~\eqref{eq:Lagrangian-details}. 
\begin{figure*}[t]
\normalsize
\setcounter{mytempeqncnt}{\value{equation}}
\setcounter{equation}{40}
\begin{align}
\label{eq:Lagrangian-details}
& \mathfrak{L}(\mathbf{F}, {\gamma}_k, \rho_k,  \bm{\lambda}_k)  =  \sum\limits_{k \in \mathcal{K}} \biggl[ V \|\mathbf{f}_{k}\|^2 - ( Q_k + A_k  + Z_k )\log_{2}(1+\gamma_k)
+  \sum\limits_{u \in \mathcal{K} \backslash k} \lambda_{u,1} |\mathbf{h}_{u}\herm \mathbf{f}_{k}|^2
\nonumber \\ & \ \ 
+ \lambda_{k,1} \biggl\{ \! \sigma_k^2 \! + \! \frac{\delta_k^2}{\rho_k} \! - \! 2 \Re \Bigl\{ \! \frac{\mathbf{f}_k\iherm \mathbf{h}_k \mathbf{h}_k\herm }{\gamma_k\iLoop} \bigl( \mathbf{f}_k \! - \! \mathbf{f}_k \iLoop \bigr) \Bigr\} \! - \! 2 \frac{|\mathbf{h}_{k}\herm \mathbf{f}_{k}\iLoop|^2}{ \gamma_k\iLoop } \! + \! \gamma_k \frac{|\mathbf{h}_{k}\herm \mathbf{f}_{k}\iLoop|^2}{ (\gamma_k\iLoop)^2 }  \! \biggr\} 
%
%
%
%
\! + \! \lambda_{k,2} \biggl\{ \! \Bigl(\frac{1}{1-\rho_k} \Bigr)  \frac{\ddot{e}_k}{\zeta_k} \! \biggr\} 
\nonumber \\ & \ \
- 2 \sum\limits_{j=1}^{K} \lambda_{j,2}  \Re \Bigl\{ \! \mathbf{f}_k\iherm \mathbf{h}_j \mathbf{h}_j\herm  \bigl(\mathbf{f}_k - \mathbf{f}_k \iLoop \bigr) \! \Bigr\}  
%
%
%
%
\! - \! \lambda_{k,2} \Bigl\{  \sum\limits_{j=1}^{K} |\mathbf{h}_{k}\herm \mathbf{f}_{j}\iLoop|^2 \! + \! \sigma_k^2 \Bigr\}  
\! + \! \rho_k \bigl\{ \lambda_{k,4} \! - \! \lambda_{k,3} \bigr\}  \! - \! \lambda_{k,4} 
\biggr].
\end{align}
%
\setcounter{equation}{\value{mytempeqncnt}}
\setcounter{equation}{41}
\vspace*{-12pt}
\hrulefill
\vspace*{-22pt}
\end{figure*}
Next, by differentiating~\eqref{eq:Lagrangian-details} with respect to primal optimization variables~$\{\mathbf{f}_k, \gamma_k, \rho_k\}$, we obtain the stationarity conditions for~\eqref{eq:P1c_InfB}, given~by
\begin{subequations}
\label{eq:KKT_Derivate}
\begin{align}
\label{eq:Derivate_Precoder}
& \nabla_{{\mathbf{f}}_k}: \mathbf{f}_k\herm \Bigl(V\mathbb{I} + \!\!  \sum\limits_{u \in \mathcal{K} \backslash k} \lambda_{u,1} \mathbf{h}_u \mathbf{h}_u\herm \Bigr) 
=  \lambda_{k,1} \frac{\mathbf{f}_k\iherm \mathbf{h}_k \mathbf{h}_k\herm }{\gamma_k\iLoop} + \mathbf{f}_k\iherm \sum\limits_{j=1}^{K} \lambda_{j,2} \mathbf{h}_j \mathbf{h}_j\herm, \\ \displaybreak[0]
\label{eq:Derivate_gamma}
& \nabla_{{\gamma}_k}: \frac{Q_k + A_k  + Z_k}{1+\gamma_k} = \lambda_{k,1} \frac{|\mathbf{h}_{k}\herm \mathbf{f}_{k}\iLoop|^2}{ (\gamma_k\iLoop)^2 }, \\ \displaybreak[0]
\label{eq:Derivate_rho}
& \nabla_{{\rho}_k}: \lambda_{k,1} \frac{\delta_k^2}{\rho_k^2} = \lambda_{k,2} \frac{\ddot{e}_k}{\zeta_k(1-\rho_k)^2} + (\lambda_{k,4} - \lambda_{k,3}).
\end{align}
\end{subequations}
Further, in addition to primal-dual feasibility, the \ac{KKT} conditions include the complementary slackness, defined~as
\begin{subequations}
\label{eq:KKT_Slackness}
\begin{align}
\label{eq:Lam1}
& \lambda_{k,1} \geq 0; \qquad \lambda_{k,1} \bigl\{  I_k(\mathbf{F}, {\rho_k}) - \widetilde{{G}}_k(\mathbf{f}_k, \gamma_k, \mathbf{f}_k\iLoop, \gamma_k\iLoop) \bigr\}=0,  \ \ \forall k, \\ \displaybreak[0]
\label{eq:Lam2}
& \lambda_{k,2} \geq 0; \qquad \lambda_{k,2} \bigl\{  \ddot{C}_k({\rho_k}) - \widetilde{\ddot{S}}_k(\mathbf{F}, \mathbf{F}\iLoop)  \bigr\}=0,  \ \ \forall k, \\ \displaybreak[0]
\label{eq:Lam3}
& \lambda_{k,3} \geq 0; \qquad \lambda_{k,3} \bigl\{ 0-\rho_k \bigr\}=0,  \ \ \forall k, \\ \displaybreak[0]
\label{eq:Lam4}
& \lambda_{k,4} \geq 0; \qquad \lambda_{k,4} \bigl\{ \rho_k - 1 \bigr\}=0,  \ \ \forall k.
\end{align}
\end{subequations}
%
%
%
%
Note that the associated dual variable~$\{\bm{\lambda}_k\}$ is strictly positive only when the constraint is tight. Specifically, setting $\lambda_{k,3}\!>\!0$ would imply $\rho_k\!=\!0$. However, in such a case, expression~\eqref{eq:I_function} would become infeasible.
Similarly, $\lambda_{k,4}\!>\!0$ results in $\rho_k\!=\!1$, and it would make expression~\eqref{eq:C_function_InfB} infeasible. 
Thus, 
we can conclude that constraints~\eqref{eq:P1c-C4_InfB} and~\eqref{eq:P1c-C5_InfB} are not tight at a feasible solution of~\eqref{eq:P1c_InfB}, i.e., $0\!<\!\rho_k\!<\!1$ and $\lambda_{k,3}\!=\!\lambda_{k,3}\!=\!0,$ $\forall k\in\mathcal{K}$. 
%
%
%
%
%
Next, let us assume the Lagrange multipliers ${\lambda}_{k,1}\!>\!0$ and ${\lambda}_{2,1}\!>\!0, \ \forall k \in \mathcal{K}$. Then, after some algebraic manipulations of expressions~\eqref{eq:KKT_Derivate} and \eqref{eq:KKT_Slackness}, the closed-form steps in the iterative method~are
\begin{subequations}
\label{eq:KKT-i iteration}
\begin{align}
\label{eq:KKT-i, f i}
& \tilde{\mathbf{{f}}}_{k} = \Bigl(V\mathbb{I} + \sum\limits_{u \in \mathcal{K} \backslash k} \lambda_{u,1}\xiLoop \mathbf{h}_u \mathbf{h}_u\herm \Bigr)^{-1}
\biggl\{ \frac{\lambda_{k,1}\xiLoop}{\gamma_k\xiLoop} \mathbf{h}_k \mathbf{h}_k\herm \mathbf{f}_k\xiLoop   +  \sum\limits_{j=1}^{K} \lambda_{j,2}\xiLoop \mathbf{h}_j \mathbf{h}_j \herm \mathbf{f}_k\xiLoop\biggr\} 
, \\ \displaybreak[0]
 \label{eq:KKT-i, f update}
& \mathbf{{f}}_{k}\iLoop = \mathbf{{f}}_{k}\xiLoop + \beta \bigl(\tilde{\mathbf{{f}}}_{k} - \mathbf{{f}}_{k}\xiLoop \bigr), \\ \displaybreak[0]
\label{eq:KKT-i, rho i}
& \rho_k\iLoop = 1 -   \frac{\ddot{e}_k}{\zeta_k} \biggl\{2 \sum\limits_{j=1}^{K} \Re \Bigl\{ \mathbf{f}_j\xiherm \mathbf{h}_k \mathbf{h}_k\herm  \bigl(\mathbf{f}_j\iLoop - \mathbf{f}_j\xiLoop \bigr) \Bigr\} 
+  \sum\limits_{j=1}^{K} |\mathbf{h}_{k}\herm \mathbf{f}_{j}\xiLoop|^2 + \sigma_k^2 \biggr\}^{-1}, \\ \displaybreak[0]
\label{eq:KKT-i, SINR i}
& \gamma_k\iLoop = 2 \gamma_k\xiLoop \!\! + \! \Biggl\{ \! 2 \Re \biggl\{ \! \frac{\mathbf{f}_k\xiherm \mathbf{h}_k \mathbf{h}_k\herm }{\gamma_k\xiLoop} \bigl(\mathbf{f}_k\iLoop - \mathbf{f}_k\xiLoop \bigr) \! \biggr\} 
\! -  \! \! \! \sum\limits_{u \in \mathcal{K} \backslash k} \! |\mathbf{h}_{k}\herm \mathbf{f}_{u}\iLoop|^2 \! - \!  \sigma_k^2 \! - \! \frac{\delta_k^2}{\rho_k\iLoop}  \! \Biggr\}  
\frac{(\gamma_k\xiLoop)^2}{|\mathbf{h}_{k}\herm \mathbf{f}_{k}\xiLoop|^2}, \\ \displaybreak[0]
\label{eq:KKT-i, lam1}
%
& \lambda_{k,1}\iLoop =  \frac{ (Q_k + A_k  + Z_k) \big(\gamma_k\xiLoop\big)^2 }{(1+\gamma_k\iLoop) |\mathbf{h}_{k}\herm \mathbf{f}_{k}\xiLoop|^2 }, \\ \displaybreak[0] 
\label{eq:KKT-i, lam2}
& \lambda_{k,2}\iLoop = \frac{\zeta_k \delta_k^2 \lambda_{k,1}\xiLoop \bigl(1 \!- \!\rho_k\xiLoop \bigr)^2 }{\ddot{e}_k \bigl(\rho_k\xiLoop\bigr)^2},
%
\end{align}
\end{subequations}
%
%
%
%
%
%
%
%
%
where $\beta$ is a positive step size, i.e., $0< \! \beta \! \leq 1$. It is worth highlighting that the SWIPT beamformer design~\eqref{eq:KKT-i, f i} inherently has a multicast structure due to \ac{EH} from RF received signals~\cite{Gautam_-2021_Journal}. Specifically, the dual variables~$\{{\lambda}_{k,1}\}$ and $\{{\lambda}_{k,2}\}$ control the balance between information unicasting and energy multicasting to each user, respectively. Note that to avoid separate outer and inner loop updates, and hence, to speed up the convergence, the fixed operating point~$\{ \mathbf{f}_k\iLoop, {\gamma}_k\iLoop \}$ is also heuristically updated in each iteration along with the associated Lagrange multipliers. Thus, in general, the monotonic convergence can not be guaranteed. 
 However, to improve the convergence behaviour, we employ the best response (BR) method to regularize the beamformer update in expression~\eqref{eq:KKT-i, f update}.
In Section~\ref{sec:sim-results}, we show via a numerical example that, with a proper choice of $\beta$ in~\eqref{eq:KKT-i, f update}, a monotonic convergence of the objective function can be achieved with a fairly small number of iterations.
%
%
%
%
The proposed method consists of iteratively solving a system of closed-form \ac{KKT} equations as summarized in Algorithm~\ref{algKKT}. 
\vspace*{-10pt}
\SetArgSty{textnormal}
\begin{algorithm}[]
\DontPrintSemicolon
	\caption{{KKT based} iterative algorithm for~\eqref{eq:P1c_InfB}} 
	\label{algKKT}
	\SetAlgoLined
	Initialize with feasible starting point $\bigl\{\mathbf{{f}}_{k}\oLoop, {\gamma}_k\oLoop, \rho_k\oLoop, \bm{\lambda}_{k}\oLoop \bigr\}$, \ $\forall k $, and set $i = 1$  \\ 
	\Repeat{convergence or for fixed iterations}{
	{
	Solve $\tilde{\mathbf{{f}}}_{k}$ from~\eqref{eq:KKT-i, f i} and update $\mathbf{{f}}_{k}\iLoop$ using~\eqref{eq:KKT-i, f update}   \\
	Compute variables $\rho_k\iLoop$, ${\gamma}_k\iLoop$ from~\eqref{eq:KKT-i, rho i}, \eqref{eq:KKT-i, SINR i}, respectively \\ 
	Obtain Lagrange multipliers~$\lambda_{k,1}\iLoop$, $\lambda_{k,2}\iLoop$ from~\eqref{eq:KKT-i, lam1}, \eqref{eq:KKT-i, lam2}, respectively  \\
	%
	%
	Set $i = i+1$
	}
}
\end{algorithm}
\vspace*{-20pt}
%
%


\vspace*{-15pt}
\section{Simulation Results}
\label{sec:sim-results}
\vspace*{-5pt}

In this section, we provide numerical examples to assess the performance of the proposed algorithms for joint transmit beamformers and receive~\ac{PS} ratios optimization.
%
%
%
%
We consider a \ac{SWIPT} system with $K\!=\!2$ \acp{UE} being served by a \ac{BS} equipped with a \ac{ULA} of $N_t\!=\!8$ antenna elements. For simplicity, we assume identical parameters for all \acp{UE}. Specifically, we set noise variance $\delta_k^2\!=\!-50$~dBm and $\sigma_k^2\!=\!-70$~dBm; \ac{EH} conversion efficiency~$\zeta_k\!=\!0.8, \ \forall t, k$~\cite{Shi-SWIPT-2014}. Without loss of generality, we assume maximum battery capacity $B_k^{\mathrm{max}}\!=\!10$~Joules; fixed energy consumption $P_k^{\mathrm{cir}}\!=\!0$~dBm; energy efficiency of decoder circuit $\vartheta_k^d\!=\!0.5$~[Joules/bit]; normalized frame duration~$T_f\!=\!1$; and scaling factor~$\omega_k~\!=~\!150, \ \forall t, k$. Further, we consider a Poisson arrival process $A_k\!\sim\!\mathrm{Pois}(\alpha)$ with allowable queue backlog $Q_k^{\mathrm{th}}\!=\!5$~bits and tolerable violation probability $\epsilon\!=\!10\%$ in problem~\eqref{eq:P1-prob}~\cite{Dileep_Globecom2020}. We set the fixed minimum harvested power requirement $\ddot{e}_k(t)\!=\!10$~dBm, $ \forall t, k$ in problem~\eqref{eq:P1b_InfB}. Finally, we set the step size $\beta = 0.25$ in expressions~\eqref{eq:KKT-i, f update}.
If not mentioned otherwise, we consider problem~\eqref{eq:P1-prob} and all the results are averaged over randomly generated 3000 channel realizations.

We consider uncorrelated Rician fading to model the radio propagation channel. Specifically, the channel vector~$\mathbf{h}_k\!\in\!\mathbb{C}^{N_t\times1}$ of $k$-th \ac{UE} consists of a deterministic Line-of-Sight~(LoS) path~$\mathbf{h}_k^{\mathrm{LoS}}$ and a spatially uncorrelated non-LoS~(NLoS) path~$\mathbf{h}_k^{\mathrm{NLoS}}$, such that %
\begin{equation}
 \label{eq:Channel}
 \mathbf{h}_k(t) = \sqrt{\frac{\kappa}{1+\kappa}}\mathbf{h}_k^{\mathrm{LoS}}(t) + \sqrt{\frac{1}{1+\kappa}}\mathbf{h}_k^{\mathrm{NLoS}}(t), \ \ \forall t, k,
\end{equation}
where $\kappa$ is the Rician factor. We set $\kappa\!=\!5$~dB unless stated otherwise. Moreover, the NLoS component is modeled independently for each time slot~$t$ using Rayleigh fading with the path-loss of $-40$~dB. Meanwhile, the LoS component follows the standard far-field \ac{ULA} model, i.e., $\mathbf{h}_k^{\text{LoS}}(t)\!=\!\sqrt{10^{-4}}[1, \  e^{-\mathrm{j}\pi\sin(\theta_k(t))}, \ldots, e^{-\mathrm{j}(N_t-1)\pi\sin(\theta_k(t))}]\tran$, where $\theta_k(t)$ is the azimuth angle of $k$-th \ac{UE} during time slot~$t$ relative to the boresight of the \ac{BS} antenna array. The azimuth angles $\theta_k(t)~\!\in~\![-\pi/2, \ \pi/2],\ \forall t, k,$ are randomly generated.

\vspace*{-16pt}
\subsection{Convergence and complexity Analysis}
\vspace*{-6pt}

First, we investigate the convergence behavior of the proposed iterative algorithms for a given randomly generated channel realization. 
Note that the solution of~\eqref{eq:P1b}, both in Algorithm~\ref{algFP} and Algorithm~\ref{algSCA}, is obtained directly from the convex optimization toolbox $\mathrm{SeDuMi}$~\cite{cvx}. {Further, the per-iteration computational complexity is dominated by the downlink beamformer optimization step and scales exponentially with the length of the beamforming vector~$N_t$.}
{It can be observed from Fig.~\ref{fig:SDR_SCA_Convergance} that the \ac{SDR}-\ac{FP} based iterative Algorithm~\ref{algFP} provides faster convergence in terms of the required approximation point updates. 
The interior point methods are usually adopted to solve SDP formulations, and each iteration requires {$\mathcal{O}\bigl((K+N_t^{2})^{3.5}\bigr)$} arithmetic operations~\cite{Qingjiang-TWC-SWIPTKey_2014, Chen-complexity}, \cite[Section III]{Sidiropoulos-Complexity}. Moreover, Algorithm~\ref{algFP} also requires additional post-approximation procedures to recover the rank-one beamforming vectors~(see Section~\ref{subsec:SDR} for more details). 
On the contrary, the \ac{SCA} based iterative Algorithm~\ref{algSCA}  attains the feasible solution without any additional steps after the convergence. The approximated convex problems can be efficiently solved as a sequence of SOCP, and the worst-case computational complexity of each iteration is {$\mathcal{O}(\bigl(K+N_t)^{3.5}\bigr)$} \cite{Shi-SWIPT-2014, Chen-complexity}.} 
{Note that both  solutions exactly coincide, implying also that the Algorithm~\ref{algFP} yields a rank-one solution for the reformulated problem~\eqref{eq:P1-prob} in this particular scenario with the above considered setup configurations.} However, the \ac{SDR}-based optimization
becomes expensive, potentially much more than \ac{SCA} framework, as the number of transmit antennas~$N_t$ at the BS increases (e.g., matrix~$\{\mathbf{F}_k\}_{\forall k}$). 
%
%
Therefore, Algorithm~\ref{algFP} can be useful for delay-constrained devices with sufficient processing capabilities, while Algorithm~\ref{algSCA} can be applied to hardware-constrained devices to produce a more computationally efficient solution.

\begin{figure}[t!]
\begin{minipage}[c]{0.49\textwidth}
 \setlength\abovecaptionskip{-0.2\baselineskip}
 \centering
 \includegraphics[trim=0.25cm 0.08cm 0.15cm 0.15cm, clip, width=1\linewidth]{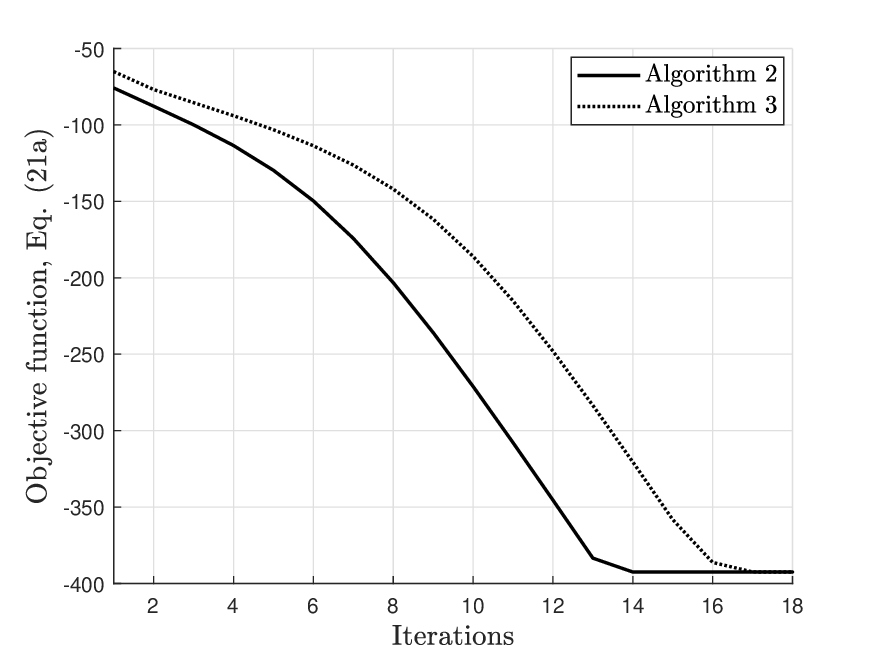}
 \caption{{Convergence behavior} of Algorithm~\ref{algFP} and Algorithm~\ref{algSCA} for~\eqref{eq:P1b} with $V\!=\!1$ and $\alpha\!=\!3$.}
 \label{fig:SDR_SCA_Convergance}
\end{minipage}
\hspace{1mm}
\begin{minipage}[c]{0.49\textwidth}
 \setlength\abovecaptionskip{-0.2\baselineskip}
 \centering
 \includegraphics[trim=0.25cm 0.08cm 0.15cm 0.15cm, clip, width=1\linewidth]{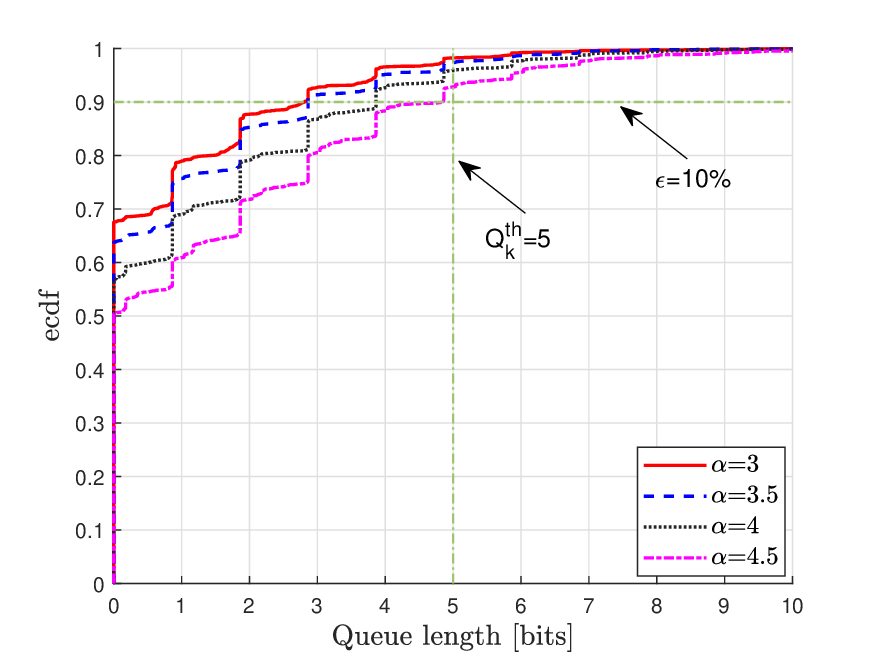}
 \caption{The queue backlog with different mean arrivals~$\alpha$ {and $V=1$}.}
 \label{fig:QueueBaclogs}
\end{minipage} 
 \vspace*{-28pt}
\end{figure}

\vspace*{-16pt}
\subsection{Impact of  parameter~$V$ and mean arrival rate~$\alpha$}
\vspace*{-6pt}

Herein, we investigate the impact of different mean arrival rate~$\alpha$ and trade-off parameter~$V$ on the achievable system performance. First, we set $V\!=\!1$ and illustrate the empirical cumulative distribution function (ecdf) of the queue backlogs in Fig.~\ref{fig:QueueBaclogs}.
It can be concluded that irrespective of the mean arrival~$\alpha$, the proposed method ensures the maximum queue backlogs of each user~$k$ (i.e., $Q_k^{\mathrm{th}}\!=\!5$) within tolerable violation probability $\epsilon\!=\!10\%$. Thus, the proposed convex relaxations to~\eqref{eq:P1-prob} still allow to satisfy the desired user-specific latency requirements (i.e., probabilistic queue backlog constraint~\eqref{eq:P1-C2} is met). {Note that similar performance trends can be obtained for different values of parameter~$V$, but this is not illustrated due to space limitations.}

\begin{figure}[t]
\begin{minipage}[c]{0.49\textwidth}
 \setlength\abovecaptionskip{-0.2\baselineskip}
 \centering
 \includegraphics[trim=0.25cm 0.08cm 0.15cm 0.15cm, clip, width=1\linewidth]{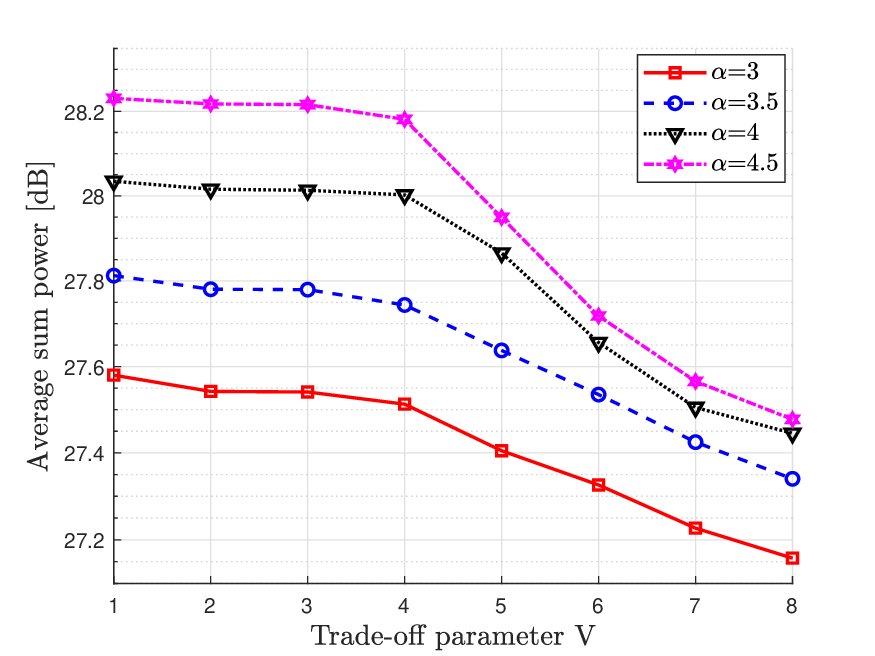}
 \caption{The BS average transmit power with {increasing $V$ and }different mean arrivals.}
 \label{fig:SumnPower}
\end{minipage}
\hspace{1mm}
\begin{minipage}[c]{0.49\textwidth}
 \setlength\abovecaptionskip{-0.2\baselineskip}
 \centering
 \includegraphics[trim=0.25cm 0.08cm 0.15cm 0.15cm, clip, width=1\linewidth]{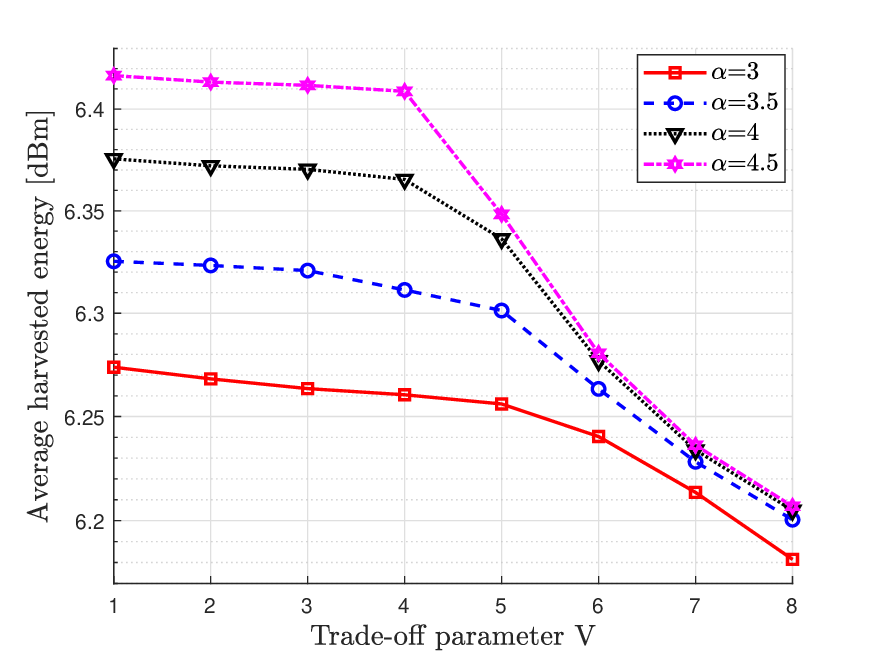}
 \caption{The receiver average harvested energy with {increasing $V$ and }different mean arrivals.}
 \label{fig:HarvestedPower_VS_V}
\end{minipage} 
 \vspace*{-28pt}
\end{figure}

Next, we examine the impact of different values of the trade-off parameter~$V$ on the average transmit power of the BS and the average harvested power at the receiver in Fig.~\ref{fig:SumnPower} and Fig.~\ref{fig:HarvestedPower_VS_V}, respectively. 
The result in Fig.~\ref{fig:SumnPower} shows that the sum-power decreases with an increasing~$V$. This is an expected behavior since~$V$ can be anticipated as a scaling factor for the sum-power objective function (see expression~\eqref{eq:P1b-objective} for details). Thus, higher values of~$V$ linearly emphasize the minimization of the BS transmit power over the sum of network queue backlogs and spare battery capacity until/unless the  queues length becomes substantially larger than the sum-power objective value. 
Furthermore, we can observe that the \ac{BS} sum-power significantly increases with the mean arrival rate. {This is mainly due to the fact that the increase in the network queue backlogs~\eqref{eq:Queue_Dynamics} enforces downlink transmission with higher data rates to reduce users' queues, e.g., at the cost of consuming more energy, and thus, requiring higher BS transmit powers to quickly recharge the devices.} Note that with the increasing~$V$ and $\alpha$, it also becomes more stringent to satisfy the user-specific latency constraints~\eqref{eq:P1-C2}, as will become clear in the following.  


Fig.~\ref{fig:HarvestedPower_VS_V} illustrates that the energy harvested at the receiver decreases with the increase of~$V$. This is mainly due to the decrease in the transmit power of the BS with an increasing~$V$, as also shown in Fig.~\ref{fig:SumnPower}. Note that a downlink transmission with less power leads, in general, to relatively less harvesting opportunities at the receivers.  
{Furthermore, we can observe that the harvested power increases with the mean arrival rate. Specifically, with the increasing network queue backlogs~\eqref{eq:Queue_Dynamics}, the BS attempts to increase the maximum downlink supported rates~\eqref{eq:Max_Data_Rate}. 
Thus, the \acp{UE} consume significant energy resources to support their current data decoding operations~\eqref{eq:Energy_Consumption}, which eventually leads to greater \ac{EH} requirements in subsequent time slots due to less available spare battery capacity (refer to Section II-C for more details). In this regard, our proposed virtual battery queue with the scaling on spare battery capacity~\eqref{eq:Pertubed_Battery} enforces more stringent \ac{EH} requirements at the UEs
to quickly recharge them
and 
also stabilize the time-average battery queues~\eqref{eq:Battery_Dynamics}.}     
All in all, there is a strong interplay among network queue length, BS transmit power, and receiver spare battery capacity, as can be seen from Fig.~\ref{fig:QueueBaclogs}$-$Fig.~\ref{fig:HarvestedPower_VS_V}.

\vspace*{-16pt}
\subsection{Network {and virtual} queue dynamics}
\label{subsec:sim-networkQueue}
\vspace*{-6pt}

Fig.~\ref{fig:ActualQueue} and Fig.~\ref{fig:VirtualQueue} respectively show the dynamics of network queue backlogs~$\{Q_k(t)\}$ and associated virtual queue~$\{Z_k(t)\}$ over time {with parameter~$V\!=\!\{1, 8\}$} for a given~\ac{UE}. 
%
%
It can be observed that the time dynamics of queues include a transient state and a steady state. Furthermore, the steady state is only attained after accumulating certain backlogs in the associated virtual queues~$\{Z_k(t)\}$ as shown in Fig.~\ref{fig:VirtualQueue}. Thus, until/unless the virtual queue reaches a certain value at which it gets saturated (e.g., {around $t\!=\!\{26, 34\}$ for $V\!=\!\{1, 8\}$}), only then the actual network queue oscillate, so as the probabilistic queue length constraint $\mathrm{Pr}\bigl\{ Q_k(t) \geq Q_k^{\mathrm{th}} \bigr\} \! \leq \! \epsilon$ is satisfied, i.e., to achieve the probabilistic delay requirements~\eqref{eq:P1-C2}. This is mainly because of the negative drift property of the Lyapunov function~\cite[Ch. 4.4]{neely2010stochastic}. Specifically, the stability of associated virtual queues~$\{Z_k(t)\}$ strictly ensures that the network queues are bounded~\cite[Theorem 2.5]{neely2010stochastic}, which allows achieving the desired user-specific queue backlogs performance.

\begin{figure}[]
\begin{minipage}[c]{0.49\textwidth}
 \setlength\abovecaptionskip{-0.2\baselineskip}
 \centering
 \includegraphics[trim=0.25cm 0.08cm 0.15cm 0.15cm, clip, width=1\linewidth]{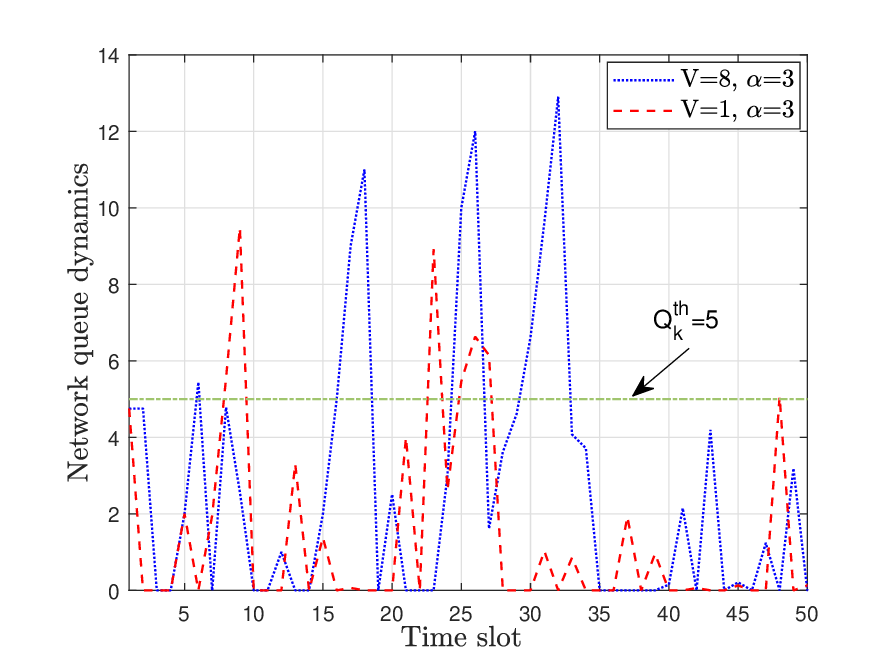}
 \caption{{The dynamics of network queues.}}
 \label{fig:ActualQueue}
\end{minipage}
\hspace{1mm}
\begin{minipage}[c]{0.49\textwidth}
 \setlength\abovecaptionskip{-0.2\baselineskip}
 \centering
 \includegraphics[trim=0.25cm 0.08cm 0.15cm 0.15cm, clip, width=1\linewidth]{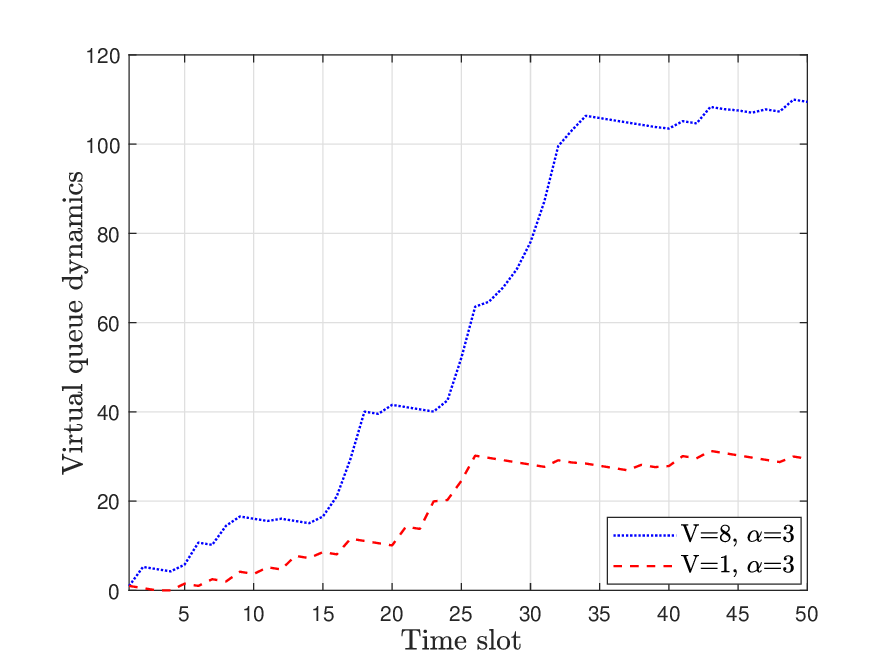}
 \caption{{The dynamics of virtual queues.}}
 \label{fig:VirtualQueue}
\end{minipage} 
 \vspace*{-30pt}
\end{figure}


\vspace*{-16pt}
\subsection{Receiver battery dynamics}
\label{subsec:sim-batteryQueue}
\vspace*{-6pt}

Next, we examine the performance of receiver battery capacity in Fig.~\ref{fig:BatteryLevels_VS_V} and Fig.~\ref{fig:BatteryECDF}. It can be concluded from Fig.~\ref{fig:BatteryLevels_VS_V} that the average battery capacity decreases with an increasing~$V$, mainly due to the decrease in BS transmit power and corresponding less harvesting opportunities at the receiver. As an example, by decreasing the parameter~$V\!=\!8$ to $V\!=\!1$, the minimum battery charge level of a given UE is improved {from $B_k\!=\!1.45$ Joules to $B_k\!=\!3.1$ Joules}, as shown in Fig.~\ref{fig:BatteryECDF}. However, the average battery capacity remains somewhat similar with different mean arrival rates.  {This is due to the fact that the arrival with a higher rate enforces higher BS transmit power, and thus more  harvested energy at the receiver (see Fig.~\ref{fig:HarvestedPower_VS_V}) to support the more energy-demanding receive decoding process.} Specifically, with increasing arrivals, the battery is quickly recharged to support the receivers' operation, and thus the battery level remains fairly stable.

\begin{figure}[]
\begin{minipage}[c]{0.49\textwidth}
 \setlength\abovecaptionskip{-0.2\baselineskip}
 \centering
 \includegraphics[trim=0.25cm 0.08cm 0.15cm 0.15cm, clip, width=1\linewidth]{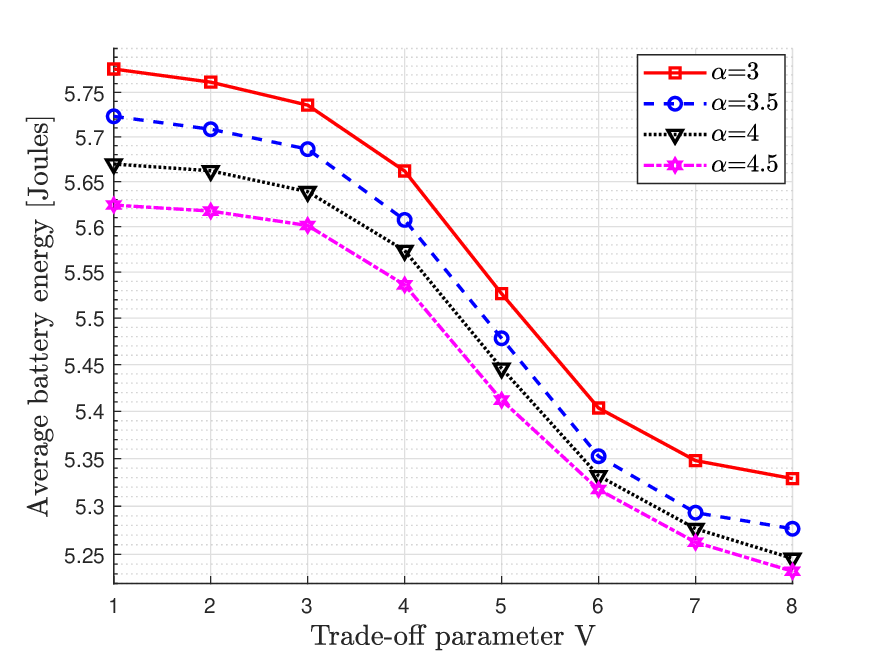}
 \caption{The receiver average battery energy with {increasing $V$ and }different mean arrivals.}
 \label{fig:BatteryLevels_VS_V}
\end{minipage}
\hspace{1mm}
\begin{minipage}[c]{0.49\textwidth}
 \setlength\abovecaptionskip{-0.2\baselineskip}
 \centering
 \includegraphics[trim=0.25cm 0.08cm 0.15cm 0.15cm, clip, width=1\linewidth]{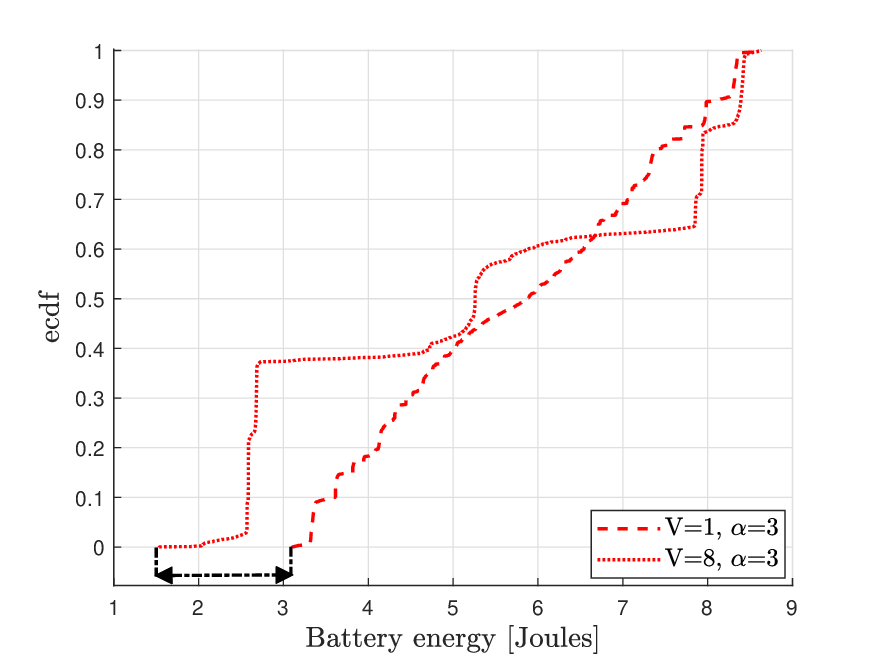}
 \caption{The ecdf of battery charge levels with $V\!\in\!\{1, 8\}$ and mean arrival~$\alpha=3$.}
 \label{fig:BatteryECDF}
\end{minipage} 
 \vspace*{-28pt}
\end{figure}



To summarize, even after the proposed relaxations to problem~\eqref{eq:P1-prob}, the solution still allows to satisfy the desired user-specific latency and maximum \ac{EH} requirements. That is, constraints~\eqref{eq:P1-C2}-\eqref{eq:P1-C4} are strictly met with the minimum BS transmit power. Thus, with the proposed dynamic control algorithm and suitable parameterization of drift-plus-penalty functions, one can easily handle the network queue backlogs, optimize the \ac{BS} transmit power, and the receivers' battery charge level, while ensuring the user-specific latency and \ac{EH} requirements of, e.g., industrial-grade delay bounded critical applications for  factory automation scenarios.



\vspace*{-16pt}
\subsection{Delay-bounded batteryless devices}
\vspace*{-6pt}

For the special scenario of delay bounded SWIPT-enabled batteryless receivers, we illustrate in Fig.~\ref{fig:SCA_KKT_Convergance} the run time and the convergence performance of Algorithm~\ref{algSCA_InfB} and Algorithm~\ref{algKKT} for solving problem~\eqref{eq:P1b_InfB}. {Note that different from the \ac{SCA} based iterative approach in Algorithm~\ref{algSCA_InfB}, the solution of Algorithm~\ref{algKKT} is based on an iterative evaluation of closed-form \ac{KKT} expressions~\eqref{eq:KKT-i iteration}, thus, it does not rely on  generic convex solvers.} 
We can observe that the solution via Algorithm~\ref{algKKT} achieves superior performance to  Algorithm~\ref{algSCA_InfB}, e.g., around $98\%$ improvement in the run time performance. {Further, the computational complexity is mainly dominated by expression~\eqref{eq:KKT-i, f i}, which  consists of matrix multiplication and inverse operations. Hence, each iteration requires {$\mathcal{O}(N_t^{2.37})$} arithmetic operations using, e.g., Coppersmith–Winograd algorithm~\cite{COPPERSMITH1990251}, \cite[Appendix C]{boyd2004convex}.} Thus, Algorithm~\ref{algKKT} provides more practical, latency-conscious, and computationally efficient implementations. 
At this point, it is worth highlighting that a \ac{KKT} based closed-form iterative method may not be attained for~\eqref{eq:P1-prob} due to underlying complexity, coupling among optimization variables, and temporally correlated user-specific battery energy and \ac{EH} constraints. 

{Finally, Fig.~\ref{fig:Queue_fixedEH} illustrates the user-specific latency performance for different minimum harvested power requirements with the trade-off parameter $V\!=\!1$ and mean arrival rate $\alpha\!=\!3$. The result shows that, irrespective of the fixed harvested power constraint~\eqref{eq:P1b-C3_InfB}, our proposed method ensures the maximum backlogs of each user~$k$ (i.e., $Q^{\mathrm{th}}_k\!=\!5$) within the allowable violation probability $\epsilon\!=\!10\%$. Thus, the proposed convex relaxations to problem~\eqref{eq:P1b_InfB} still allow to satisfy the desired latency requirements.}
%
The other numerical examples are omitted
due to space limitation, but similar performance trends can easily be obtained for
~\eqref{eq:P1b_InfB}~\cite{dileep_ISWCS, dileep_PIMRC}.

\begin{figure}[t!]
\begin{minipage}[c]{0.49\textwidth}
 \setlength\abovecaptionskip{-0.2\baselineskip}
 \centering
 \includegraphics[trim=0.25cm 0.08cm 0.15cm 0.15cm, clip, width=1\linewidth]{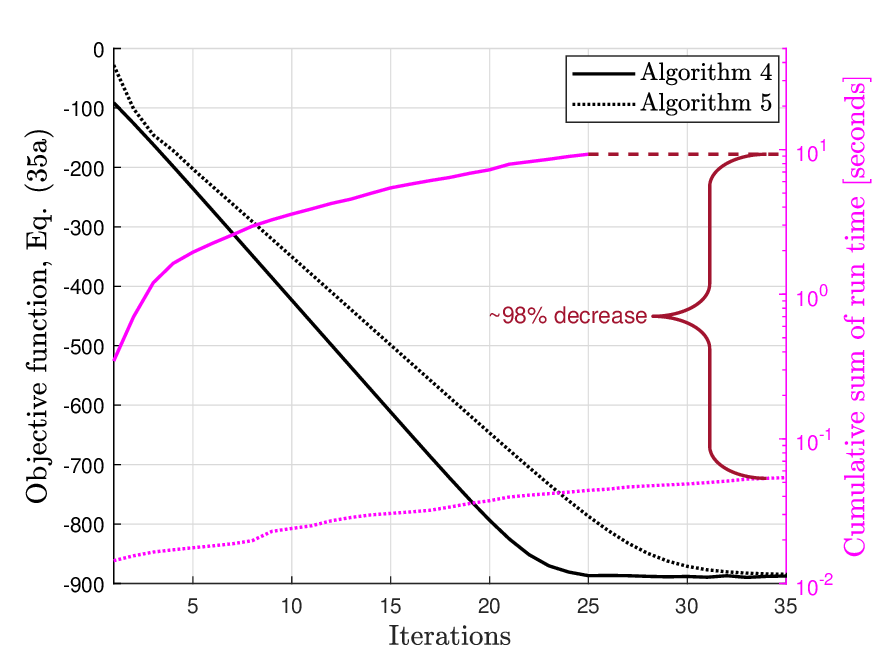}
 \caption{{Convergence behavior} of Algorithm~\ref{algSCA_InfB} and Algorithm~\ref{algKKT} for~\eqref{eq:P1b_InfB} with $V\!=\!1$ and $\alpha\!=\!3$.}
 \label{fig:SCA_KKT_Convergance}
\end{minipage}
\hspace{1mm}
\begin{minipage}[c]{0.49\textwidth}
 \setlength\abovecaptionskip{-0.2\baselineskip}
 \centering
 \includegraphics[trim=0.25cm 0.08cm 0.15cm 0.15cm, clip, width=1\linewidth]{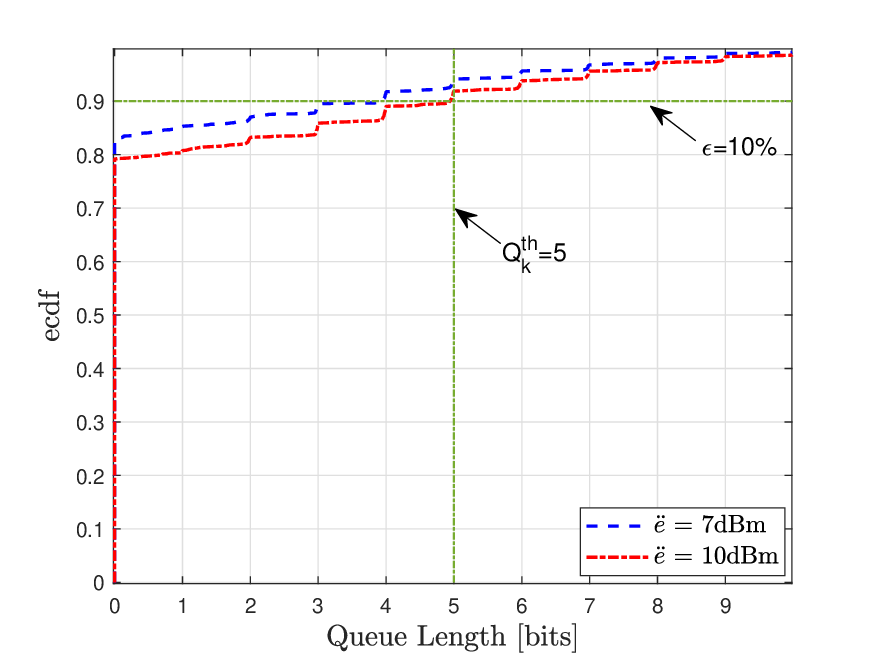}
 \caption{{The queue backlog performance for different EH requirement with $V\!=\!1$ and $\alpha\!=\!3$.}}
 \label{fig:Queue_fixedEH}
\end{minipage} 
 \vspace*{-20pt}
\end{figure}

\section{Conclusion}
\label{sec:Conclusion}

In this paper, we considered a \ac{PS} based \ac{SWIPT} system and provided a joint optimization of transmit beamforming vectors and receive \ac{PS} ratios, while accounting for the limited battery energy at each \ac{UE}. Specifically, we considered a long-term average \ac{BS} transmit power minimization problem concurrently satisfying the user-specific latency and maximum \ac{EH} requirements. The proposed radio resource allocation schemes efficiently avoid the receivers' battery depletion phenomenon by preemptively incorporating the spare battery capacity and EH fluctuations in a time dynamic mobile access network. To provide a tractable solution, we employed the Lyapunov optimization framework, and provided an online dynamic control algorithm to obtain a series of per-time slot deterministic subproblems. 
Furthermore, the coupled and non-convex constraints were handled by applying the techniques of \ac{SDR}-\ac{FP} and \ac{SCA} framework. A closed-form iterative algorithm was designed by solving a system of \ac{KKT} optimality conditions for a special case of delay bounded batteryless \acp{UE}. The simulation results manifested the robustness of the proposed design to realize an energy-efficient \ac{SWIPT} system for industrial-grade 
delay bound 
applications. 

\bibliographystyle{IEEEtran}
\bibliography{IEEEabrv,ref_conf_short,ref_jour_short,referencesSWIPT}

\begin{IEEEbiography}[{\includegraphics[trim=0.042cm 0.042cm 0.042cm 0.042cm, width=1in,height=1.25in,clip,keepaspectratio]{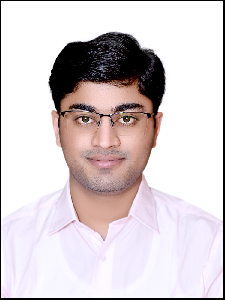}}]{Dileep Kumar} (Graduate Student Member, IEEE) received the master's degree in communication engineering from the Indian Institute of Technology Bombay, India, in 2015, and Dr.Sc. (Tech.) degree in communications engineering from the University of Oulu, Finland, in 2022. From 2015 to 2017, he worked for {NEC} Corporation, Tokyo, Japan, as a Research Engineer. In 2018, he joined the Centre for Wireless Communications (CWC), University of Oulu, Finland. His research interest includes signal processing for wireless communication systems.
\end{IEEEbiography}

\begin{IEEEbiography}[{\includegraphics[width=1in,height=1.25in,clip,keepaspectratio]{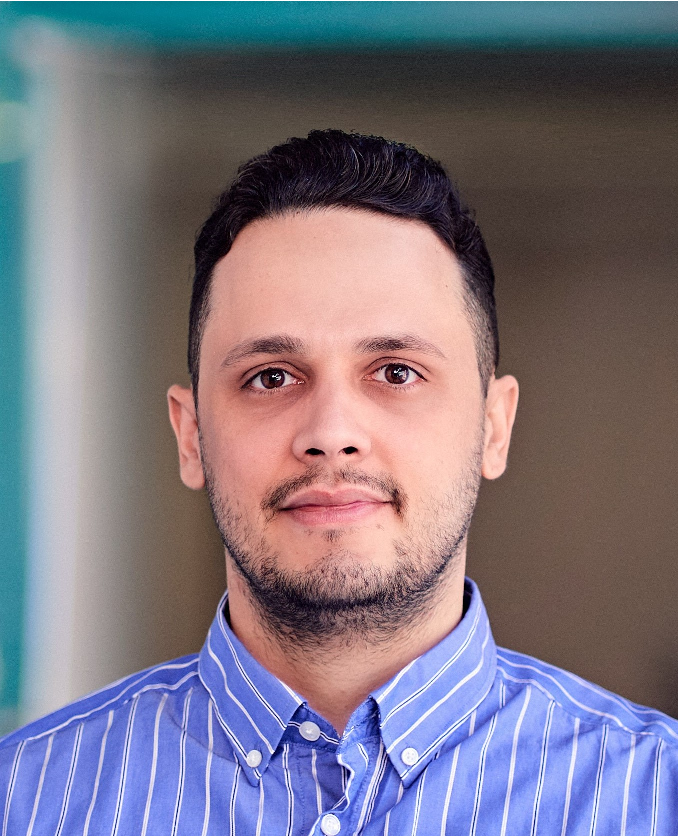}}]{Onel L. A. López} (Member, IEEE) was born in Sancti-Spíritus, Cuba, in 1989. He received the B.Sc. (1st class honors, 2013), M.Sc. (2017) and D.Sc. (with distinction, 2020) degree in Electrical Engineering from the Central University of Las Villas (Cuba),  the Federal University of Paraná (Brazil) and the University of Oulu (Finland), respectively. From 2013-2015 he served as a specialist in telematics at the Cuban telecommunications company (ETECSA). He is a collaborator to the 2016 Research Award given by the Cuban Academy of Sciences, a co-recipient of the 2019 IEEE European Conference on Networks and Communications (EuCNC) Best Student Paper Award, the recipient of both the 2020 best doctoral thesis award granted by Academic Engineers and Architects in Finland TEK and Tekniska Föreningen i Finland TFiF in 2021 and the 2022 Young Researcher Award in the field of technology in Finland. He authored the book entitled ``Wireless RF Energy Transfer in the massive IoT era: towards sustainable zero-energy networks'', Wiley, Dec 2021. He currently holds an Assistant Professorship (tenure track) in sustainable wireless communications engineering at the Centre for Wireless Communications (CWC), Oulu, Finland. His research interests include sustainable IoT, energy harvesting, wireless RF energy transfer, wireless connectivity, machine-type communications, and cellular-enabled positioning systems.
\end{IEEEbiography}

\begin{IEEEbiography}[{\includegraphics[width=1in,height=1.25in,clip,keepaspectratio]{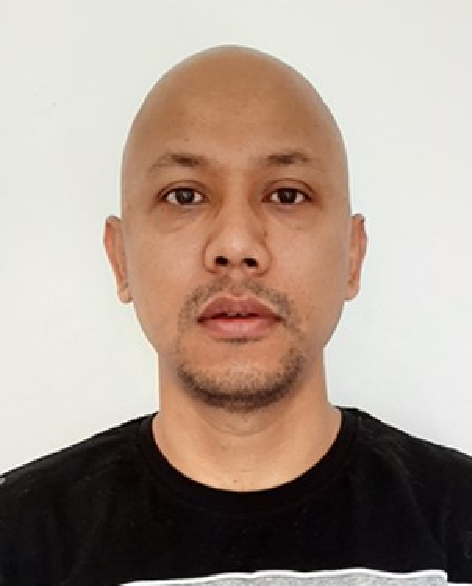}}]{Satya Krishna Joshi} (Member, IEEE) received the {M.Eng.} degree in telecommunications from the Asian Institute of Technology, Pathumthani, Thailand, in 2007, and {Ph.D} degree from the University of Oulu, Finland in 2018. His research interests include application of optimization techniques for signal processing and {MIMO} communications.
\end{IEEEbiography}

\begin{IEEEbiography}[{\includegraphics[width=1in,height=1.25in,clip,keepaspectratio]{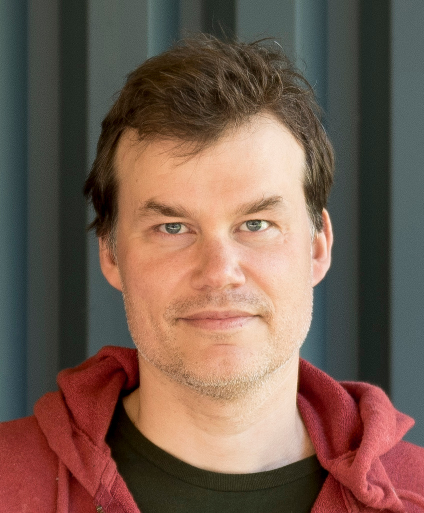}}]{Antti T\"{o}lli} (Senior Member, IEEE) is a Professor with the Centre for Wireless Communications (CWC), University of Oulu. He received the Dr.Sc. (Tech.) degree in electrical engineering from the University of Oulu, Oulu, Finland, in 2008. From 1998 to 2003, he worked at Nokia Networks as a Research Engineer and Project Manager both in Finland and Spain. In May 2014, he was granted a five year (2014-2019) Academy Research Fellow post by the Academy of Finland. During the academic year 2015-2016, he visited at EURECOM, Sophia Antipolis, France, while from August 2018 till June 2019 he was visiting at the University of California Santa Barbara, USA. He has authored numerous papers in peer-reviewed international journals and conferences and several patents all in the area of signal processing and wireless communications. His research interests include radio resource management and transceiver design for broadband wireless communications with a special emphasis on distributed interference management in heterogeneous wireless networks. From 2017 to 2021 he served as an Associate Editor for IEEE Transactions on Signal Processing. 
\end{IEEEbiography}

\end{document}

%% file: packagesDK.tex
\usepackage{cite}

\usepackage[printonlyused]{acronym} 

\usepackage[cmex10]{amsmath} 
\usepackage{amsthm} 
\usepackage{amssymb} 
\usepackage{algorithmic}
\usepackage{array}
\usepackage{cancel} 

\usepackage[caption=false,font=footnotesize]{subfig}
\usepackage[pdftex]{graphicx}
\graphicspath{{./figsWIPT/}}
\DeclareGraphicsExtensions{.pdf,.jpeg,.png}

\usepackage{color}

\usepackage{balance} 
\usepackage{bm} 
\usepackage{caption} 

\usepackage{threeparttable} 
\usepackage{multicol} 
\usepackage{multirow} 
\usepackage{rotating} 

\usepackage{upgreek} 
\usepackage{url} 
\usepackage{verbatim} 

\usepackage{mathtools} 
\usepackage{mathrsfs} 

\usepackage{xfrac} 

\usepackage{fancyhdr}
\usepackage{amsmath}

\usepackage{etoolbox}

\usepackage{optidef}
\usepackage[ruled,linesnumbered]{algorithm2e}

%% file: newcomDK.tex
\newcommand{\herm}{^{\mbox{\scriptsize H}}}

\newcommand{\tran}{^{\mbox{\scriptsize T}}}

\newcommand{\oLoop}{^{{\scriptsize (0)}}}
\newcommand{\iLoop}{^{{\scriptsize (i)}}}

\newcommand{\xiLoop}{^{{\scriptsize (i-1)}}}

\newcommand{\iherm}{^{\scriptsize(i){\mbox{\scriptsize H}}}}
\newcommand{\xiherm}{^{\scriptsize(i-1){\mbox{\scriptsize H}}}}

\newcommand{\HypTst}[2]{\raise2.5ex\hbox{\scriptsize$#1$} \hspace{-1.3em}\displaystyle\gtreqless\hspace{-1.2em}\raise-1.1em\hbox{\scriptsize$#2$}}

\newcommand{\cref}[1]{C\ref{#1}}

%% file: acronsDK.tex
\acrodef{3-G}{3-Generation}
\acrodef{3GPP}{3rd Generation Partnership Project}
\acrodef{a.k.a}{also-known-as}
\acrodef{AAS}{active antenna system}
\acrodef{ADMM}{alternating direction method of multipliers}
\acrodef{AoA}{angle-of-arrival}
\acrodef{AoD}{angle-of-departure}
\acrodef{AHB}{Abel hybrid bound}
\acrodef{A-TS}{alternating Taylor\'s series}
\acrodef{A-LASSO}{\emph{adaptive}-least absolute shrinkage and selection operator}
\acrodef{AWGN}{additive white Gaussian noise} 
\acrodef{AGC}{Automated Gain Control}
\acrodef{BER}{Bit Error Rate}
\acrodef{BB}{Bhattacharyya Bound}
\acrodef{BF}{beamformer}
\acrodef{BFGS}{Broyden-Fletcher-Goldfarb-Shanno}
\acrodef{BS}{base-station}
\acrodef{BBU}{baseband processing unit}
\acrodef{BSI}{Battery State Information}

\acrodef{cdf}{cumulative distribution function}
\acrodef{CEP}{Circular Error Probability}
\acrodef{CoMP}{Coordinated Multi-Point}
\acrodef{CIR}{Channel Impulse Response}
\acrodef{C-RDM}{Cross-Ranging Direction Matrix}
\acrodef{C-NLS}{Constrained Non-linear Least Squares}
\acrodef{CRLB}{Cram\'{e}r-Rao Lower Bound}
\acrodef{CSIT}{Channel State Information at the Transmitter}
\acrodef{CSI}{channel state information}
\acrodef{CP}{Cyclic-Prefix}
\acrodef{CB}{Coordinated Beamforming}

\acrodef{DC}{Distance Contraction}
\acrodef{DCT}{Distance Contraction Theory}
\acrodef{DE}{Distance Error}
\acrodef{DFT}{Discrete Fourier Transform}
\acrodef{DL}{Downlink}
\acrodef{DoA}{Direction-of-Arrival}
\acrodef{DSSS}{Direct Sequence Spread Spectrum}
\acrodef{DoF}{degree-of-freedom}
\acrodef{DPS}{Dynamic Point Selection}

\acrodef{EDF}{Euclidean Distance Function}
\acrodef{EDM}{Euclidean Distance Matrix}
\acrodef{EFIM}{Equivalent Fisher Information Matrix}
\acrodef{EHF}{Extremely High Frequency}
\acrodef{EKT}{Euclidean Kernel Transformation}
\acrodef{EK}{Euclidean Kernel}
\acrodef{ERII}{Equivalent Ranging Information Intensity}
\acrodef{ETSI}{European Telecommunications Standards Institute}
\acrodef{EH}{energy harvesting}
\acrodef{FIM}{Fisher Information Matrix}
\acrodef{FCC}{Federal Communications Commission}
\acrodef{FDD}{Frequency-Division-Duplex}
\acrodef{FP}{fractional programming}

\acrodef{GLE}{Geometric-constrained Location Estimation}
\acrodef{G-LS}{Global Least Squares}
\acrodef{GNSS}{Global Navigation Satellite System}
\acrodef{GP}{Gaussian Process}
\acrodef{GP-LVM}{Gaussian Process Latent Variable Model}
\acrodef{GPS}{Global Positioning System}
\acrodef{GDC}{Global Distance Continuation}
\acrodef{GDOP}{Geometric Dilution of Precision}
\acrodef{GTRS}{Generalized Trust Region Subproblems}
\acrodef{HCRB}{Hammersley-Chapmann-Robbins Bound}
\acrodef{HB}{Hybrid Bound}
\acrodef{HOSVD}{High-Order Singular-Value-Decomposition}
\acrodef{HPBW}{half power beamwidth}
\acrodef{iid}{independent identically distributed}
\acrodef{ICT}{Information Communication Technology}
\acrodef{I-EKT}{Inverse Euclidean Kernel Transformation}
\acrodef{IoT}{Internet-of-Things}
\acrodef{IIoT}{Industrial Internet-of-Things}
\acrodef{iff}{if and only if}
\acrodef{IFT}{Inverse Fourier Transform}
\acrodef{ITU}{International Telecommunication Union}
\acrodef{I/N}{Interference-to-Noise}
\acrodef{ID}{information decoding}
\acrodef{JT}{Joint Transmission}
\acrodef{JT-CoMP}{Joint Transmission Coordinated Multi-Point}

\acrodef{KKT}{Karush-Kuhn-Tucker}
\acrodef{KLT}{Karhunen-Lo\'{e}ve Transform}
\acrodef{LASSO}{least absolute shrinkage and selection operator}
\acrodef{LARSO}{Least Absolute Residual and Selection Operator}
\acrodef{LBSs}{Location Based Services}
\acrodef{LM}{Levenberg-Marquardt}
\acrodef{LQ}{Link-Quality}
\acrodef{LS}{Least Squares}
\acrodef{LLS}{Linearized Least Squares}
\acrodef{LT}{Location-Tracking}
\acrodef{LTE}{Long Term Evolution}
\acrodef{LoS}{line-of-sight}
\acrodef{LOESS}{Locally weighted Scatterplot Smoothing}
\acrodef{LP}{Linear Programming}
\acrodef{LHS}{ left-hand side}

\acrodef{MAP}{Maximum A Posteriori}
\acrodef{MOO}{Multiple-Objective Optimization}
\acrodef{MOS}{Mean Opinion Score}
\acrodef{MIMO}{multiple-input-multiple-output}
\acrodef{MU-MIMO}{multi-user multiple-input multiple-output}
\acrodef{MU-MISO}{multi-user multiple-input single-output}
\acrodef{MMSE}{Minimum Mean Squared Error}
\acrodef{mmWave}{millimeter-wave}
\acrodef{MSE}{mean square error}
\acrodef{ML}{Maximum Likelihood}
\acrodef{MS}{mobile-station}
\acrodef{MDS}{Multidimensional Scaling}
\acrodef{MSPE}{Mean-Squared-Position-Error}
\acrodef{MUSIC}{MUltiple SIgnal Classification}
\acrodef{MRT}{maximum ratio transmission}
\acrodef{MURM}{Minimum User-Rate Maximization}

\acrodef{N-DC}{Negative-Distance Contraction}
\acrodef{NLOS}{non-line-of-sight}
\acrodef{NLS}{Non-Linear-Least Squares}
\acrodef{NP-hard}{Non-deterministic Polynomial-time hard}
\acrodef{NPRM}{Notice of Proposed Rulemaking}
\acrodef{OFDM}{orthogonal frequency-division multiplexing}
\acrodef{OMP}{Orthogonal Matching Pursuit}
\acrodef{OTDoA}{observed-Time-Difference-of-Arrival}
\acrodef{psd}{positive semi-definite}
\acrodef{PSD}{power spectral density}
\acrodef{pdf}{probability density function}
\acrodef{PDP}{Power Delay Profile}
\acrodef{P-DC}{Positive-Distance Contraction}
\acrodef{PE}{Position Error}
\acrodef{PEB}{Position Error Bound}
\acrodef{PREB}{position-rotation error bound}
\acrodef{PSM}{Positive Semi-definite Matrix}
\acrodef{POCS}{Projection On Convex Sets}
\acrodef{PCA}{Principal Component Analysis}
\acrodef{PPCA}{Probabilistic Principal Component Analysis}
\acrodef{PS}{power splitting}
\acrodef{QoD}{Quality-of-Design}
\acrodef{QoL}{Quality-of-Location}
\acrodef{QoS}{Quality-of-Service}
\acrodef{QoE}{Quality-of-Experience}
\acrodef{QoI}{Quality-of-Information}
\acrodef{R95}{95$\%$ Radius}
\acrodef{RBF}{Radial Basis Function}
\acrodef{RDM}{Ranging Direction Matrix}
\acrodef{REB}{Rotation Error Bound}
\acrodef{RF}{radio frequency}
\acrodef{R-GDC}{Range-Global Distance Continuation}
\acrodef{R-LS}{Regularized Least-Squares}
\acrodef{RSS}{Received Signal Strength}
\acrodef{RSSI}{Received Signal Strength Index}
\acrodef{RMB}{Reuven-Messer Bound}
\acrodef{RMSE}{root-mean-squared-error}
\acrodef{RII}{Ranging Information Intensity}
\acrodef{RR}{Ridge-Regression}
\acrodef{RRU}{remote radio unit}
\acrodef{RHS}{ right-hand side}

\acrodef{SCE}{Sparse Channel Estimator}
\acrodef{SCM}{spatial channel model}
\acrodef{SDP}{semidefinite programming}
\acrodef{SIMO}{Single-Input-Multiple-Output}
\acrodef{SISO}{Single-Input-Single-Output}
\acrodef{SLAM}{Simultaneous Localization and Mapping}
\acrodef{SMACOF}{Stress-of-a-MAjorizing-Complex-Objective-Function}
\acrodef{SQP}{Sequential Quadratic Programming}
\acrodef{SR-LS}{Squared-Range Least-Squares}
\acrodef{SR-GDC}{Square Range-Global Distance Continuation}
\acrodef{SER}{Symbol-Error-Rate}
\acrodef{SB}{Stochastic Bound}
\acrodef{SNR}{Signal-to-Noise Ratio}
\acrodef{SHF}{Super High Frequency}
\acrodef{SINR}{signal-to-interference-plus-noise ratio}
\acrodef{SD}{Steepest-Descent}
\acrodef{SA}{Simulated-Annealing}
\acrodef{SR}{Squared Range}
\acrodef{SCA}{successive convex approximation}
\acrodef{SOCP}{second-order cone program}
\acrodef{SWIPT}{simultaneous wireless information and power transfer}
\acrodef{SDR}{semidefinite relaxation}
\acrodef{SWIPT}{simultaneous wireless information and power transfer}
\acrodef{TDoA}{Time-Difference-of-Arrival}
\acrodef{TS}{Taylor's Series}
\acrodef{ToF}{Time-of-Flight}
\acrodef{ToA}{Time-of-Arrival}
\acrodef{TDD}{Time-Division-Duplex}
\acrodef{TTI}{Transmission Time Interval}
\acrodef{ULA}{uniform linear array}
\acrodef{UMTS}{Universal Mobile Telecommunications System}
\acrodef{URA}{uniform rectangular array}
\acrodef{UWB}{Ultra-WideBand}
\acrodef{UHF}{Ultra High Frequency}
\acrodef{UE}{user equipment}

\acrodef{WLS}{Weighted Least Squares}
\acrodef{WC}{Weighted Centroid}
\acrodef{WSRM}{Weighted Sum-Rate Maximization}

\acrodef{ZC}{Zadoff-Chu}
\acrodef{ZF}{Zero-Forcing}